\documentclass[a4paper,onecolumn,11pt,accepted=2023-02-22]{quantumarticle}
\pdfoutput=1

\input{include}

\usepackage[pass, showframe]{geometry}

\newgeometry{margin=1.0in,includefoot,footskip=30pt}

\begin{document}

\title{\text{Quantum Multi-Solution Bernoulli Search} 
\text{with Applications to Bitcoin’s Post-Quantum Security}}

\author{Alexandru Cojocaru}
\affiliation{University of Maryland}
\email{cojocaru@umd.edu}
\orcid{0000-0003-2710-1170}
\author{Juan Garay}
\email{garay@cse.tamu.edu}
\orcid{0000-0003-0366-7110}
\affiliation{Texas A\&M University}
\author{Aggelos Kiayias}
\affiliation{University of Edinburgh and IOHK}
\email{akiayias@inf.ed.ac.uk}
\author{Fang Song}
\affiliation{Portland State University}
\email{fang.song@pdx.edu}
\orcid{0000-0002-3098-6451}
\author{Petros Wallden}
\affiliation{University of Edinburgh}
\email{petros.wallden@ed.ac.uk}
\orcid{0000-0002-0255-6542}

\maketitle

\begin{abstract}
A proof of work (PoW) is an important cryptographic construct 
which enables a party to convince other parties that they have invested some effort in solving a computational task.  
Arguably, its main impact has been in the setting of cryptocurrencies such as Bitcoin and its underlying blockchain protocol,
which have received significant attention in recent years due to its potential for various applications as well as for solving fundamental distributed computing questions in novel threat models.
PoWs enable the linking of blocks in the blockchain data structure, and thus the problem of interest is the feasibility of obtaining a sequence  (``chain'') 
of such proofs.

At the same time, the rapid development in quantum computing makes the threats to cryptography more and more concerning. 
In this work, we examine the hardness of finding such \emph{chain} of PoWs against quantum strategies.
We prove that the chain of PoWs problem reduces to a problem we call {\em multi-solution Bernoulli search}, for which we establish its quantum query complexity. Effectively, this is an extension of a threshold direct product theorem to an \emph{average-case} unstructured search problem. Our proof, adding to active recent efforts, simplifies and generalizes the recording technique due to Zhandry [Crypto 2019].

As an application, we revisit the formal treatment of security of the core of the Bitcoin consensus protocol, called the {\em Bitcoin backbone} [Eurocrypt 2015], in a setting where the adversary has quantum capabilities while the honest parties remain classical, and show that the protocol's security
holds under a quantum analogue of the classical ``honest majority'' assumption that we formulate.
Our analysis indicates that the security of the Bitcoin backbone protocol is guaranteed provided that the number of adversarial quantum queries is bounded so that each quantum query is worth $O(p^{-1/2})$ classical ones, where $p$ is the probability of success of a single classical query to the protocol's underlying hash function.
Somewhat surprisingly, the wait time for safe settlement of transactions in the case of quantum adversaries matches (up to a constant) the safe settlement time in the classical case.

\end{abstract}

\tableofcontents

%%%%%%%%%%%%%%%%%%%%%%%%%%%%%

% 1. INTRODUCTION

%%%%%%%%%%%%%%%%%%%%%%%%%%%%%

\section{Introduction} \label{sec:intro}

A {\em proof of work} (\pow)
enables a party to convince other parties that considerable effort has
been invested in solving a computational task.  First introduced by
Dwork and Naor~\cite{DN92}, $\pows$ serve as candidate solutions to
thwarting spam emails and denial-of-service attacks. More recently in
the setting of cryptocurrencies, $\pows$ have proven indispensable and
manifested profound impact, where they lay the basis for consensus
protocols in permissionless blockchains. In particular,
Bitcoin~\cite{Nakamoto2009} and its underlying blockchain protocol
have drawn remarkable attention for the potential to resolve
fundamental distributed computing problems in various threat models,
as well as to enable novel applications.

In the blockchain setting, the objective of a \pow{} is to confirm new
transactions to be included in the blockchain.
To successfully create a \pow{} in Bitcoin, one needs to find a value
(``witness'') such that evaluating a hash function (SHA-256) on this
value together with (the hash of) the last block and new transactions
to be incorporated, yields an output below a threshold. A party who
produces such a \pow{} gets to append a new block to the blockchain
and is rewarded. A {\em blockchain} hence consists of a sequence of
such \emph{blocks}. Each party maintains such a blockchain, and
attempts to extend it via solving a \pow. We refer to the task of
creating a chain of multiple blocks as the \emph{\chainpow{}}
problem. In more detail, we define these blockchain-inspired problems
as follows.

\paragraph{Definition \textnormal{(Blockchain \pow---Informal)}.}
\textit{Given a hash function $h$, a positive integer $T$, and a string $z$
  representing the hash value of the previous block, the goal is to
  find a value $ctr$ such that:}
\[ h(ctr, z) \le T \, .\]

\paragraph{Definition \textnormal{(\chainpow---Informal)}.}
\textit{Given a family of hash functions
  $h_i : X \times Y \rightarrow X$, an initial value $x_0 \in X$, the
  goal is to output a \emph{chain of length $k$}
  consisting of $y_0,\ldots, y_{k-1} \in Y$ such that:}
\[ x_{i+1}: = h_{i}(x_{i},y_{i}) \text{ and } x_{i+1} \leq T \text{
    for all } i \in \{0, \cdots, k-1 \} \, .  \]

Intuitively, the hardness of the \chainpow{} problem is crucial to a
robust blockchain and trustworthy Bitcoin applications. However, due
to the complex nature of these protocols, obtaining a firm security
claim in a formal mathematical model turns out to be challenging. At
Eurocrypt 2015, Garay {\em et al.}~\cite{GKL15} developed an
abstraction of Bitcoin's underlying protocol termed the ``Bitcoin
backbone,'' which 
inspired a series of other formal treatments of Bitcoin (e.g.,~\cite{PSS16,GKL17,BMTZ17}). In essence, the backbone
abstraction of~\cite{GKL15} models the hash function of the \pow{} as
a random oracle (RO)~\cite{DBLP:conf/ccs/BellareR93}, and it assumes a
\emph{uniform} configuration among a fixed (albeit unknown) number of
parties (``miners''), a fraction of which may behave arbitrarily as
controlled by an adversary. Namely, the parties are endowed with the
same computational power, as measured by the permitted number of
queries to the RO per round. Following a modular approach, essential
security properties of the backbone protocol, such as \emph{common
  prefix} and \emph{chain quality}, are formulated, which are proven
sufficient to realize blockchain applications, notably a robust public
transaction ledger (a.k.a. ``Nakamoto consensus''). These properties,
relying crucially on the hardness of \chainpow, can then be
established assuming \emph{honest majority} of computational power.

The formal treatment in the works above assumes that parties are
instantiated by ``classical'' computers. However, we are witnessing a
rapid development in quantum computing bringing concerning threats to
cryptographic protocols. For example, it is known that quantum
computers can cause devastating breaks to both public key
cryptography~\cite{Shor97} and symmetric-key
cryptography~\cite{KLLNP16,SS17}. In addition, the unique features of
quantum information, such as intrinsic randomness and no-cloning,
render many classical security analyses (e.g.,
rewinding~\cite{vdG97,Watrous09,Unruh12}) obsolete, and even the right
\emph{modeling} of security can be elusive (e.g.,
\cite{HSS15,AGM18,Zhandry13}). In particular, existing analyses of
blockchain protocols are conducted in the random oracle model. When
quantum attackers are present, Boneh {\em et al.}~\cite{BDF+11} argued
the need for granting the attackers the ability to query the random
oracle in \emph{quantum superposition}. This gives rise to the
\emph{quantum random oracle} (QRO) model, which is far more
challenging to reason about (see Section~\ref{sec:rel} for an
account). Several fundamental questions need to be revisited, such as:

\begin{quote}
  \emph{What is the complexity of $\chainpow$ in the quantum query
    model?
    Can we formally establish the
    quantum security of the Bitcoin backbone protocol, by,
    for example, 
    modifying
    the framework of~\cite{GKL15} to work against quantum
    adversaries in the QRO model?}
\end{quote}

In the remainder of this section we summarize our contributions,
provide a technical overview, and discuss closely related works. The
full presentation of our results then follows.

%==============================%
\subsection{Our Contributions}
\label{intro:result}
%==============================%

In this work, we prove the quantum query complexity of $\chainpow$,
and then formally establish the quantum security of the Bitcoin
backbone protocol in the QRO model.

First, consider a family of independent and uniformly random functions
$\{h_i: X \times Y \to X\}$, given as quantum unitary black-boxes. Let
$p: = T/{|X|}$ be the fraction of target outputs in the co-domain,
which in the context of Bitcoin represents the \emph{difficulty} of
mining. We prove the following:

\begin{theorem}[Informal]\label{thm:chain_pow_intro}
For any quantum adversary $\mathcal{A}$ having $N$ quantum queries, the probability that $\mathcal{A}$ solves the \chainpow{} problem by outputting a chain of size at least $k$, is at most
\[ P(N, k) \leq \left( e^2\left({N/k}+1\right)^2p \right)^k. \]
\end{theorem}

We show this by considering a variant in which the solution does not need
to form a chain. This amounts to searching multiple inputs in a
non-uniform random function, and its quantum query hardness
immediately implies the one for \chainpow, which we derive
precisely. More specifically, we have a Boolean function whose output
is sampled according to a Bernoulli distribution \emph{independently}
for each input, and we wish to find multiple preimages of $1$ (call it
\mrsearch). We prove the quantum query complexity of this search
problem by extending Zhandry's elegant recording
technique~\cite{Zhandry19} to point-wise independent
\emph{non-uniform} distributions. In fact, this establishes a quantum
\emph{average-case} {strong direct-product theorem} (SDPT)
for this family of non-uniform search problems\footnote{It is
  actually a \emph{threshold} SDPT, which is usually stronger. The
  distinction is not essential to our discussion, and interested
  readers are referred to~\cite{LR13}.}.

Turning to the quantum security of the Bitcoin backbone protocol, we
examine all the components of the~\cite{GKL15} framework in the presence of classical honest parties against
quantum adversaries in the QRO model, and our guiding principle is to ``lift''
it to the quantum setting with \emph{as few} changes as possible. One
crucial change in our analysis is to reformulate a central analytical
tool in~\cite{GKL15} named \emph{typical execution}, since in the
presence of quantum adversaries it becomes ill-defined. Assisted by
our quantum complexity bound for \chainpow, we are able to identify a
quantum analogue of the \emph{honest majority} condition, under which
the desired security properties of the Bitcoin backbone protocol
follow. As a side benefit, the resulting analysis makes the reduction
to $\chainpow$ explicit, which further improves the degree of
modularity of the classical analysis.

We elaborate on our results and how they are obtained in the technical
overview (Section~\ref{sec:tech-overview}). Before that, we highlight
some interesting aspects of our results and put them in perspective.

\paragraph{Quantum query complexity of $\chainpow$ and $\mrsearch$.}
Roughly speaking, the formal query bounds we show match the intuitive
argument based on Grover's quantum search algorithm, which at first glance may not
look surprising. Instead, we view the following values potentially
more impactful. 

\begin{tiret}
\item Our proof provides another working example demonstrating the
  power and potential of the recording technique. While the conceptual
  idea behind the original technique is remarkably neat, the
  \emph{execution} is often a lot more complex, which in turn
  dictates the success or
  failure of applying the technique. This is also the case in many
  recent generalizations. Our execution of a particular non-uniform
  recording technique is carefully specified and ``packaged'' in clean
  modules, and could be treated as a template to derive other proofs
  based on it.

  We note that the basic idea of non-uniform extensions is natural and
  has already been investigated (e.g.,~\cite{AMRS20}). In particular,
  in independent work, Hamoudi and Magniez~\cite{HM21} consider
  essentially the same search problem as we do (\mrsearch), proving a quantum
  query bound by extending Zhandry's technique\footnote{The
    formulation of the search problem in~\cite{HM21} is in a slightly
    restricted form, but their analysis should in principle work for
    the general version we consider here. We also need strengthened
    bounds to derive tighter security for the Bitcoin backbone
    protocol. See Appendix~\ref{sec:comparison} for further
    discussion.}. In fact, if the sole goal is to show a query bound
  for the search problem, not necessarily as tight as what we obtain,
  it is plausible that one can simulate non-uniform functions by
  uniform functions and then resort to the original uniform recording
  technique (cf.~\cite{LZ19_mult}). Again, the merit
  would lie in the execution,
  and in any case a stand-alone treatment of a non-uniform case
  remains beneficial.
  
\item As also noted in~\cite{HM21}, the quantum query complexity of
  the search problem establishes an \emph{average-case} strong direct
  product theorem (SDPT) via the recording technique, which is
  combinatorial in nature. SDPTs are highly desirable in a host of
  theoretical models~\cite{Unger09,KSdW07,Sherstov12,LR13}, and it
  basically assert that the success probability of solving multiple
  instances of a problem drops exponentially unless investing equal
  multiples of the resource for solving one instance. In the context
  of quantum query complexity, SDPT has been extensively investigated,
  where the \emph{polynomial method}~\cite{BBH+01} and the
  \emph{(generalized) adversary method}~\cite{Ambainis02} are the two
  dominating approaches to prove it. The polynomial method builds upon
  \emph{analytical} properties of polynomials, while the (``modern'')
  adversary method often resorts to advanced \emph{algebraic} tools
  such as representation theory. In contrast, the combinatorial
  approach based on the recording technique, which usually follows
  more elementary and intuitive arguments, could open a new route to
  establishing a SDPT. Besides, existing examples of SDPT proven by
  either method (polynomial or adversary method) are typically for
  \emph{worst-case} problems and usually difficult to extend to the
  \emph{average-case}. This for instance is the case for many examples
  of SDPT for the standard unstructured
  search~\cite{grover_optimality,BBGH05,ASdW09,Ambainis10}. Our new
  approach via the recording technique seems especially well suited to
  reason about {\em average-case} problems and also to induce hardness
  bounds needed in the cryptographic setting.

\end{tiret}

\paragraph{Quantum security of the Bitcoin backbone protocol.} Once
the quantum hardness of $\chainpow$ is established, it would be
tempting to ``plug'' it into the classical framework of~\cite{GKL15}
in order to derive its quantum security. However, we stress that, as
more and more works demonstrate, when quantum adversaries are present
every link of the provable-security framework needs to be reexamined,
including the attack models, security goals, and security
reductions.

Specifically, it can be justified that the security goals of the
Bitcoin backbone protocol (e.g., common prefix, chain quality) may
stay unchanged fortunately. But as soon as the attacks come into play,
we observe a crucial simplification assumption in~\cite{GKL15} that
bears no clear quantum analogue. In a nutshell, classically it is
without loss of generality to assume that an adversary, just as honest
users, solves every \pow{} simply by querying the random oracle on a
sequence of inputs. As a result, the adversarial behavior is fixed and
the protocol can be described by an explicit random experiment, in
which we can conduct probabilistic analysis of various events such as
those indicating secure configurations of the system. This is how the
analysis proceeds in~\cite{GKL15}, by formulating a central notion of
\emph{typical executions}, which enables establishing the desired
security properties under an honest-majority condition. In our case,
since we can no longer assume a specific adversarial strategy that
includes all possible quantum attacks, the above simplification
becomes unsound.

Our solution introduces a simple remedy which allows inheriting the
classical analysis to a large extent. It boils down to an alternative
and stronger characterization of typical executions for the quantum setting. Although this
characterization is not as intuitive as in the classical case, it
enables the analysis of quantum adversaries and proof that the backbone
protocol's security goals under a post-quantum honest-majority
condition we identify. In addition, as we mentioned earlier, the
security of the protocol is explicitly reduced to the hardness of
\chainpow. In turn, this leads to a more modular analysis even in the classical
setting, and might help in the quantum security analysis of other
blockchain-based protocols~\cite{GKL17,GKLP16,GKP19}.

Finally, we remark that
the focus of this work is the simplified scenario in which only the adversary is quantum-capable, while the honest parties remain classical. We consider this analysis a first step towards the more general scenario, that also has independent interest in its own right since it address a realistic practical situation that the near-term development of quantum technologies may bring us. Specifically, imagine that large quantum computers are not available to the public, due to their price and required technologies, hence most users participating in the Backbone Bitcoin  would only own classical computing devices. Large quantum computers would only be held by a few big companies or states, where most of them would not be willing or interested in participating in the bitcoin blockchain. The question we address is what would happen if one of those very few but powerful players decided to attempt to compromise the security of a PoW-based blockchain. 
The ultimate aim would be to address the full quantum setting, in which 
there could be an arbitrary number of quantum-capable players (both honest and adversarial), and this remains an important direction for future research.

%%%%%%%%%%%%%%%%%%%%%%%%%%%%%%%%%%%%%%%%%%
\subsection{Technical Overview}
\label{sec:tech-overview}
%%%%%%%%%%%%%%%%%%%%%%%%%%%%%%%%%%%%%%%%%%

%%%%%%%%%%%%%%%%%%%%%%%%%%%%%%%%%%%%%%%%%%%%%%%%%%%%%%%%%%%%%%%%%%%%%%%%%%%%%%%%%%%%
\subsubsection{Quantum Query Complexity of Multi-solution Bernoulli Search and \texorpdfstring{\\}{} \chainpow}
\label{intro:chainpow}
%%%%%%%%%%%%%%%%%%%%%%%%%%%%%%%%%%%%%%%%%%%%%%%%%%%%%%%%%%%%%%%%%%%%%%%%%%%%%%%%%%%%

To show the quantum query complexity of the \chainpow{} problem
(Theorem~\ref{thm:chain_pow_intro}), we consider a simpler variant
(for the adversary), called \mrsearch, which relaxes the
requirement that the \pows{} need to form a chain. Bounding the
adversary's performance in the simpler problem readily provides a
bound for the \chainpow{} problem.

Hence, we focus on the quantum query complexity of \mrsearch, and prove
the main bound on the success probability for solving it.

\paragraph{Definition (\mrsearch).} \textit{ Given a Boolean function
  $f$ as a black-box, such that each input $x$ is independently
  assigned the value $f(x)=1$ with probability $p$, find $k$ distinct
  preimages of $\, 1 \,$ for $f$.}

\begin{theorem}[Informal]
\label{thm:kbersearch-informal}
For any quantum adversary $\mathcal{A}$ having $N$ quantum queries, the probability that $\mathcal{A}$ solves \mrsearch{} is bounded by:
\begin{equation}
    p_k^N \leq 4(1-p) p^k \left( \sum_{i = 0}^k (\sqrt{1 - p})^i \cdot \binom{N}{i} \right)^2  \leq 
    \frac{1}{k}
    \cdot \left(\frac{eN\sqrt{p}}{k}\right)^{2k}.
  \end{equation}
\end{theorem}

Our proof of Theorem~\ref{thm:kbersearch-informal} relies on the recording technique of quantum queries in the
quantum random oracle model due to Zhandry~\cite{Zhandry19}, designed for uniformly sampled functions. We
modify the framework to accommodate functions sampled according to
Bernoulli distributions. In Zhandry's recording technique, the key
observation is that the oracle holds a private register with the
(uniform) superposition over all possible functions. Then if we view the oracle register in the Fourier domain (i.e., applying a Fourier transformation on it), each query modifies {\em exactly one position} (from $0$ to $1$, or the other way around).

In our setting, the oracle's state will now be a Bernoulli superposition of the form $ \sum_f \sqrt{\alpha_f} \ket{f}$, where $\alpha_f$ is the probability to sample the function $f$ according to a Bernoulli distribution. The view between the standard and Fourier domains in the uniform does not hold anymore. Instead, we consider a different dual domain, and extend the primal-dual perspective to Bernoulli distributions. Switching between the primal and dual domain is based on a  unitary $U_p$, which essentially rotates around the Z axis depending on the probability parameter $p$, which specializes to the Hadamard transformation in the uniform case when $p=1/2$. As a result, in the dual domain we start with the oracle state being initialized in the all-$0$ state, and each adversarial quantum query will rotate exactly one of its positions (in superposition) using the $U_p$ transformation. Consequently, after $N$ quantum queries we will have at most $N$ rotated qubits in the oracle register and the rest of the qubits will remain $0$. Intuitively, this gives us a bound on the adversary's knowledge by examining how many non-zero entries are in the oracle's register.

Having established a relation between solving the \mrsearch{} problem and the oracle's property of having at least $k$ non-zero entries in its register, we just need to bound the probability of the latter event.
The final step is achieved by bounding the ``progress'' made in enhancing the probability amplitude of the desired oracle's states after each query.

From the obtained bound on the \mrsearch{} problem, we can also derive a bound on the harder variant, the \chainpow{} problem.

\begin{lemma}[Informal]
The probability of solving \chainpow{} using $N$ quantum queries is at most the probability of solving \mrsearch{} using $N + k$ quantum queries.
\end{lemma}

%%%%%%%%%%%%%%%%%%%%%%%%%%%%%%%%%%%%%%%%%%%%%%%%%%%%%%%%%%%%%%%%%%%%%%%%%%%%%%%%%%%%
\subsubsection{Bitcoin Backbone Security against Quantum Adversaries}
\label{intro:bbqs}
%%%%%%%%%%%%%%%%%%%%%%%%%%%%%%%%%%%%%%%%%%%%%%%%%%%%%%%%%%%%%%%%%%%%%%%%%%%%%%%%%%%%

We assume a single quantum adversary who runs a computation up to
depth $Q$ per round, i.e., that the adversary makes at most $Q$
superposition queries sequentially in each round to the QRO.
When we consider $s$ consecutive rounds, we denote $N = sQ$ the total
number of queries to the QRO in $s$ rounds, which is also the depth of
the quantum computation in $s$ rounds\footnote{One can also think of a scenario where
the adversary controls multiple quantum processors and the overall depth is smaller (see end of Section~\ref{sec:pqbca}).}. We ``lift'' the classical
analytical framework of \cite{GKL15} to the quantum setting in the
following steps:
\begin{newenum}
\item Define the notion of \emph{post-quantum typical execution} of a
  blockchain in the presence of honest players and a quantum
  adversary.  This is a critical step towards modularizing the
  analysis;
\item show that under a post-quantum typical execution, the two
  desired properties, common prefix and chain quality, follow for
  suitable choices of parameters;
\item identify a post-quantum \emph{honest majority} condition under
  which a post-quantum typical execution occurs with high probability.
\end{newenum}

From our main theorem regarding the
query complexity of the \chainpow{} problem
(Theorem~\ref{thm:chain_pow_intro}), we can bound two relevant
quantities: (i) the expected length of an adversarial chain, and (ii)
the (overwhelming) probability that this bound
holds.
The former leads to a post-quantum ``honest majority'' condition,
where we require that the total number of quantum queries $Q$ per
round of the attacker, has to be less than the total number of
classical queries of all the honest parties divided by an extra
$O(p^{-1/2})$ factor, where $p$ is the probability of success of a
single query and, informally, represents the difficulty level of the
\pow. To get this condition, we examine closely the requirements for a typical execution and ensure that the honest chain exceeds by a small constant factor the largest chain that adversaries can make with non-negligible probability.

The latter determines the wait time for safe settlement, which intuitively represents the number of rounds $s$ that need to pass
in order to ensure that the common prefix and chain quality properties
hold except with the same negligible probability as in the classical
case.  Our analysis indicates that the required number of rounds
matches, up to a constant, the number of rounds required in the
classical adversarial setting.  This somewhat surprising result
implies that for post-quantum security, the number of ``block
confirmations'' necessary for a transaction to be accepted in order to
protect against double-spending, is almost the same as in the
classical setting.

Finally, we emphasize that our work, apart from setting different
parameters, preserves the generality of the classical analysis of
\cite{GKL15}. For instance, it captures strategies correlated with
honest-parties' actions as well as long-term attacks (such as selfish
mining~\cite{eyalsirer2014}).

\paragraph{Interpretation of results for Bitcoin security.}
We can now provide a comparison between
the analysis of the Bitcoin backbone protocol against classical adversaries of \cite{GKL15} and 
the current analysis against general quantum adversaries.
First off, let us summarize the parameters in the Bitcoin backbone analysis: \\

\begin{table}[H]
\small
\begin{center}
\begin{tabular}{ |l l| }
\hline
 $n$: \# honest parties & $q$: \# honest classical queries per round \\
\hline
 $Q$: \# adversarial quantum queries per & $f$: prob. at least one honest party \\ 
 round & generates a \pow{} in a round \\
\hline
 $\epsilon$: concentration quality of random variables & $\kappa$: security parameter \\
\hline
 $k$: \# blocks for common prefix & $\mu$: chain quality parameter \\
\hline
 $s$: \# rounds & $p$: prob. of success of a single \\
 (We refer to Section~\ref{sec:model_defs} for the definition of round.)  &  classical query  \\
\hline
 \end{tabular}
\end{center}
\caption{\sl Parameters used in our analysis.}
\label{tab:notation}
\end{table}

Table~\ref{table_comp} shows the comparison of results obtained in the two adversarial settings. The relevant 
conditions and quantities are:
\begin{tiret}
\item ``Honest Majority,'' which expresses the relation between the honest hashing power and the (classical or quantum) hashing power of the adversary; 
the expected number of adversarial blocks in $s$ consecutive rounds;
\item the probability of a ``typical execution'' (i.e., 
the probability that the required bounds on the number of adversarial queries hold); and
\item the number of rounds required to reach the same level of security as in the classical adversaries setting. \\
\end{tiret}

\begin{table}[H]
\begin{center}
\setlength\tabcolsep{4.0pt}
 \begin{tabular}{||c | c | c ||} 
 \multicolumn{1}{c}{}  &   \multicolumn{1}{c}{Classical Adversary} &  \multicolumn{1}{c}{General Quantum Adv. (this work)} \\ 
 \hline\hline
$\begin{array} {ccl} \text{Honest} \\ \text{Majority} \end{array}$ & $\begin{array} {lcl} \frac{t}{n - t} < 1 - 3(f + \epsilon) \ \ \text{\cite{GKL15}} \\ \text{$t = $ number (classical) adv. } \end{array}$ &  $Q \leq \frac{(1 - \epsilon)f(1 - f)}{(1 + \epsilon)e \sqrt{p}}$\\ 
 \hline
$\begin{array} {ccl} \text{Maximum Expectation} \\ \text{of Adversarial \pows} \end{array}$ & $pqt \cdot s$ \ \ \ \ \text{\cite{GKL15}} & $\leq (1+\epsilon)\sqrt{e^2p}\cdot Q\cdot s$ \\ 
 \hline
$\begin{array} {ccl} \text{ Probability of} \\ \text{Concentration} \end{array}$  & $\begin{array} {lcl} {P}_{\text{cl}} = 1 - e^{-\Omega(\epsilon^2fs)} \ \  \text{\cite{GKL15}} \end{array}$ & $P_{\text{q}} = 1 - e^{-\Omega((1 - \epsilon) f (1 - f)s} $\\
 \hline
 $\begin{array} {ccl} \text{ Number of} \\ \text{Rounds} \end{array}$ & $s_{\text{cl}}$ & $s_{\text{q}} = O\left( \frac{\epsilon^2}{(1 - \epsilon)(1 - f)} \cdot s_{cl}\right)$ \\
\hline
\end{tabular}
\end{center}
 \caption{\sl Classical vs quantum adversaries' bounds.}
 \label{table_comp}
\end{table}

From Table~\ref{table_comp}
we highlight two main aspects. First, noting that the probability that at least one
honest party generates a \pow{} in a round is $f = npq$, (we emphasize that $p$ is much smaller than the honest hashing power, implying that $f = npq < 1$), the current post-quantum honest 
majority condition can be expressed as $Q \lessapprox n\cdot q \cdot p^{1/2} \cdot O(1)$, 
meaning that each quantum query is worth $p^{-1/2}$ classical queries. 
Second, the number of rounds for safe settlement is (up to a constant) the same as in
the classical case.

%%%%%%%%%%%%%%%%%%%%%%%%%%%%%%%%%%%%%%%%%%
\subsection{Related Work} \label{sec:rel}
%%%%%%%%%%%%%%%%%%%%%%%%%%%%%%%%%%%%%%%%%%

A first step towards understanding Bitcoin's vulnerabilities against
quantum attacks was taken by Aggarwal {\em et al.}~\cite{ABLS18}. They
pointed out the imminent break of the elliptic-curve-based signature
scheme in Bitcoin, and argued that in contrast the \pow{} is
relatively resistant to near-term quantum computations due to their
slow clock speed and large overhead of quantum error
correction. 
The Bitcoin protocol in a setting where honest parties are also
quantum was considered in~\cite{LRS18,sattath2018insecurity}. While
these papers offer interesting observations, they fall short of formal
security guarantees. Applications of smart contracts in the quantum world are explored in \cite{CS20}, where the authors propose a classical-quantum payment system which can scale better than blockchains by relying on quantum money techniques.

A number of proof techniques have been developed for the QRO over the
years.
For example, one can simulate a quantum random
oracle~\cite{Zhandry12_prf,Zhandry15,SY17}, program it under a variety
of circumstances~\cite{ES15,Unruh15}, establish generic security of
hash functions~\cite{HRS16,BES18,HS19,LZ19_mult}. These techniques
enable proving the quantum security of many cryptographic schemes in
the QRO model ~\cite{DHK17,SXY18,AHU19,LZ19_fs,DFMS19}. Zhandry's
recording technique has inspired many follow-up with various
improvements and new applications
(e.g.,~\cite{KSS+20,CGLQ20,KKPP20,Cza21}).

A related but different problem than \chainpow{}, called proofs of
sequential work, was analyzed in~\cite{CFHL21,BLZ20}. A central task
there is to find a $q$-chain, which is a sequence
$(x_0, x_1, ..., x_q)$ such that $x_i = H(x_{i-1})$ for
$1 \leq i \leq q$, using fewer than $q$ queries. This strict relation
between the number of quantum queries and the size of the solution
marks a drastic distinction from our search problem in the Bitcoin
setting.

An interesting generalization of the recording technique to a parallel
query model is investigated in the two papers
above~\cite{CFHL21,BLZ20} as well as in~\cite{Unruh21}. We want to
point out that this generalization however would not be sufficient in
the Bitcoin context in its current form. One immediate challenge is
that the parties can act in an adaptive and heterogeneous way, and one
needs to take into account classical communication between parallel
parties.

As a final note, in a previous version of this work~\cite{CGKSW19}, we
analyzed a restricted family of quantum-attack strategies. In the
current version, we fully resolve this limitation using completely
different techniques. A more thorough discussion is deferred to
Appendix~\ref{sec:comparison}.

%%%%%%%%%%%%%%%%%%%%%%%%%%%%%

% 2. PRELIMINARIES

%%%%%%%%%%%%%%%%%%%%%%%%%%%%%

\section{Preliminaries} \label{sec:prelim}

%%%%%%%%%%%%%%%%%%%%%%%%%%%%%%%%%%%%%%%%%%%%%%%%%%%%%%%%%%%%%%%
 \subsection{The Recording Oracle Technique} \label{subsec:COT}
%%%%%%%%%%%%%%%%%%%%%%%%%%%%%%%%%%%%%%%%%%%%%%%%%%%%%%%%%%%%%%%

The recording oracle technique of Zhandry \cite{Zhandry19} is a powerful tool that allows keeping track of a quantum adversary's knowledge when interacting with a \textit{uniform} random oracle. 

Assume that the underlying uniformly sampled function of the quantum random oracle is $f : \{0, 1\}^m \rightarrow \{0, 1\}$. This is typically modelled as 
queries to an oracle performing the map: $\ket{x}\ket{y} \rightarrow \ket{x}\ket{y \oplus f(x)}$.\\
Alternatively, this is equivalent to queries to a ``phase'' oracle performing the map: 
$\ket{x}\ket{y} \rightarrow (-1)^{y \cdot f(x)} \ket{x}\ket{y}$.

The starting point is the following essential observation: an adversary querying the quantum oracle cannot distinguish between the following two settings. In the first setting the function is fixed and sampled uniformly at random. In the second setting, there exists an extra function register, prepared in uniform superposition over all possible functions $\frac{1}{\sqrt{2^{2^m}}}\sum_f \ket{f}$, and all further actions/operations are conditional on this register\footnote{Here, we abuse the notation and denote by $f$ the $2^m$-bit string representing the truth table of the function $f$.}.
The function register is not accessible to the adversary and, mathematically, we can trace-it out, resulting exactly to a uniform (classical) mixture over all possible functions.

Using this observation, we will consider that when the adversary performs a query $\sum_{x, y} a_{x, y} \ket{x, y}$, the entire state of the system before the query (composed of adversary's and oracle's state) is of the form: $ \sum_{x, y} a_{x, y} \ket{x, y} \otimes \frac{1}{\sqrt{2^{2^m}}} \sum_f \ket{f}$. Then the state of the system after the query (using the phase oracle) becomes: $ \ket{\phi_1} := \sum_{x, y} a_{x, y} \ket{x, y} \otimes \frac{1}{\sqrt{2^{2^m}}} \sum_f \ket{f} (-1)^{y \cdot f(x)}$. This shows that the phase resulted from the query can equally be viewed as affecting the oracle's state. 

Specifically, we will next analyze how each query affects the oracle's state. We first denote the $2^m$-bit string $P_{(x, y)}$ as the string having value $y$ on position $x$ and $0$ everywhere else. Using this notation, we can express the state after the query: $ \ket{\phi_1} := \sum_{x, y} a_{x, y} \ket{x, y} \otimes \frac{1}{\sqrt{2^{2^m}}} \sum_f \ket{f} (-1)^{ \langle f , P_{(x, y)} \rangle}$, where $\langle f, P_{(x, y)} \rangle$ denotes the inner product of the two strings. \\
This can also be generalized to obtain the state of the system after $t$ queries is of the form: \\
$\ket{\phi_t} := \sum_{x_1,...,x_t, y_1,...,y_t} a_{x_1, ..., x_t, y_1, ..., y_t} \ket{\psi_{x_1, ..., x_t, y_1, ..., y_t}} \otimes \frac{1}{\sqrt{2^{2^m}}} \sum_f \ket{f} (-1)^{\langle f, (P_{(x_1, y_1)} \oplus ... \oplus P_{(x_t, y_t)}) \rangle}$. \\
But now, by applying the Hadamard operation to all the qubits of the final register (oracle's state), the oracle state would become: 
$\ket{P_{(x_1, y_1)} \oplus ... \oplus P_{(x_t, y_t)}}$. \\
Consequently, we notice that in the Fourier domain the evolution of the adversary's queries would be the following:
We start with the all-zero quantum state in the oracle register; namely, the state of the oracle is denoted with $D = 0^{2^m}$. Then, with each adversarial query we will XOR $P_{(x, y)}$ to the string $D$. Finally, we notice that as a result, after $t$ queries the oracle's state $D$ will have at most $t$ bits equal to $1$.

%%%%%%%%%%%%%%%%%%%%%%%%%%%%%

% 3. K-BERSEARCH

%%%%%%%%%%%%%%%%%%%%%%%%%%%%%

%%%%%%%%%%%%%%%%%%%%%%%%%%%%%%%%%%%%%%%%%%%%%%%%%%%%%%%%%%%%%%%%%%%%%%%%%%%%%%%%%%%%
\section{Query Complexity of \mrsearch  \label{sec:main_analysis}} %%%%%%%%%%%%%%%%%%
%%%%%%%%%%%%%%%%%%%%%%%%%%%%%%%%%%%%%%%%%%%%%%%%%%%%%%%%%%%%%%%%%%%%%%%%%%%%%%%%%%%%

In this section we prove the quantum query complexity of the
\emph{randomized} search problem denoted as \mrsearch{}. This
generalizes existing direct product theorems
(cf.~\cite{ASdW09,KSdW07,Sherstov12,LR13}) to an average-case setting.

\begin{protocol}
\textbf{Problem $\mrsearch$: {$k$ solutions randomized search}} \\

\noindent\textbf{Given}: $f: \bit^m \to \bit$ generated by $\samp$: for
each $x\in \bit^m$, \emph{independently} set

\[ f(x) = \left\{ \begin{matrix} 
1, & \text{with probability } p\\
0, & \text{otherwise} \\
\end{matrix}\right. \]

\noindent\textbf{Goal}: Find $x_1,\ldots,x_k$ such that $f(x_i) = 1$
for all $i \in [k]$.
\end{protocol}

\begin{theorem} \label{lemma:bound_hard} For any quantum adversary
  making at most $N$ queries, the probability $p_k^N$ of solving
  $\mrsearch$ satisfies:
  \[ p_k^N \leq 4(1-p) p^k \left( \sum_{i = 0}^k (\sqrt{1 - p})^i
      \cdot \binom{N}{i} \right)^2 \le \frac{1}{\pi k} \cdot
    \left(e^2(N/k)^2p\right)^{k} \, . \]
\end{theorem}

The expression following the second inequality comes from some
elementary algebraic simplification, whose proof can be found in
Appendix~\ref{app:proof_corr_kbersearch}. In the following sections we
prove the bound in the first expression. We first develop a framework
of the recording technique for functions drawn according to a
point-wise independent Bernoulli distribution, and then apply it to
\mrsearch.

%%%%%%%%%%%%%%%%%%%%%%%%%%%%%%%%%%%%%%%%%%%%%%%%%%%%%%%%%%%%%%%%%%%%%%%%%%%%%%%%%%%%
\subsection{Recording Technique Framework for Bernoulli Distributions}    \label{sec:COT_bern}
%%%%%%%%%%%%%%%%%%%%%%%%%%%%%%%%%%%%%%%%%%%%%%%%%%%%%%%%%%%%%%%%%%%%%%%%%%%%%%%%%%%%

At a high level, we proceed in the following steps:
\begin{newenum}
    \item[1.] Extend the primal-dual framework to functions sampled according to Bernoulli distributions:
    \begin{enumerate}
        \item First define a proper dual domain to accommodate for functions sampled according to this distribution;
        \item Determine the mapping between the primal and dual domains;
        \item Show how each quantum query affects the dual domain;
        \item Finally, map back to the standard domain and analyse what is the form of the system state in the standard domain;
    \end{enumerate}
    \item[2.] For the specific \mrsearch{} problem, relate the probability of finding a solution to the structure of the final state in the standard domain:
    \begin{enumerate}
        \item Relate solving $\mrsearch$ to causing the desired property in the oracle (i.e., at least $k$ non-zero entries in the visited positions). 
        \item Bound the probability of the occurrence of desired property in the oracle for any $N$-query algorithm.
    \end{enumerate}
\end{newenum}

We recall from Section~\ref{subsec:COT} that the key observation in
the uniform setting is to keep the random function coherent, i.e., a
uniform superposition of the truth table of all possible functions
$ \sum_{f \in \{0, 1\}^{2^m}}\ket{f}$. Then if we look at it in the
Fourier domain, this superposition becomes simply $\ket{0^{2^m}}$. In
this viewpoint, a query $|x,y\rangle$ to the oracle has the effect
that it modifies (in superposition) one position $x$ of the truth
table. Along the course of the query algorithm, the truth table gets
updated and an algorithm's knowledge of the oracle translates to
combinatorial properties of the truth table.

In our case, we extend this primal-dual framework to a non-uniform
distribution. Loosely speaking, we will see that solving $\mrsearch$
amounts to that at least $k$ of the entries in the truth table being
rotated towards $\ket{1}$.

\begin{definition}[Bernoulli random function and Bernoulli superposition]
We call $f: \bit^m \to \bit$ a Bernoulli random function, if for every input $x$, $f(x)$ is sampled according to a Bernoulli distribution. We will denote by $\alpha_f$ the probability of $f$ being chosen according to our Bernoulli sampling procedure. For each $f$, we denote its truth table also by $f$, which is a $M:=2^m$ bit string, where the $x$'th bit corresponds to the function value $f(x)$. We can hence prepare a coherent oracle state $\sum_{f} \sqrt{\alpha_f} \ket{f}_F$, also called a Bernoulli superposition, in a function register $F$, which we call the standard Bernoulli oracle (\stdbo). Clearly, measuring it in the standard basis would sample a function $f$ according to our Bernoulli distribution.
\end{definition}

Recall in the standard query model, an oracle query is modelled as: $\ket{x,y} \mapsto (-1)^{y \cdot f(x)} \ket{x}\ket{y}$. The standard Bernoulli oracle action can be described as follows:

\begin{definition}[Standard Bernoulli Oracle Query]
An oracle query in $\stdbo$ is modelled as:
\begin{equation}
  \stdbo: \ket{x,y} \otimes \sum_f \sqrt{\alpha_f} \ket{f}_F \mapsto
  \ket{x, y} \otimes \sum_f \sqrt{\alpha_f} (-1)^{y \cdot f(x)} \ket{f}_F
\end{equation}
\end{definition}

From the perspective of any quantum query algorithm $A$, it is indistinguishable whether it accesses $\stdbo$ or a sampled Bernoulli random function $f$. Namely:
\[ \Pr[A^f(\cdot) = 1: f\gets\samp  ]  = \Pr[A^\stdbo (\cdot) = 1] \, .\]

We describe a map which will allow us to switch between the standard
and dual domains. Note that in the special case of a uniform random
function, it becomes the Hadamard gate ($p=1/2$).

\begin{definition}[Map Standard-Dual domains]
For general $p$, we consider the following unitary $U_p$ acting on a single qubit:
\begin{equation}
 U_p: \ket{b} \mapsto \sqrt{1 - p} \ket{b} + (-1)^b \sqrt{p}\ket{b\oplus 1}   
\end{equation}

Then let $\beru: = \otimes_{x\in \bit^m} U_p^x$ where $U_p^x = U_p$ for all $x$. We observe that:
\begin{equation}
   \beru \ket{0^M}_F = \sum_f \sqrt{\alpha_f} \ket{f}_F \, . 
\end{equation}

\end{definition}

We can now examine the knowledge of an algorithm about the underlying
oracle.  Consider running any query algorithm $A$ with respect to the
standard domain. Intuitively, by our earlier observation, at the end
of the execution, the function register will contain superposition of
strings $D$ of bounded Hamming weight, representing modified truth
tables in the dual domain. The non-zero positions (in these strings
$D$) would entail what the query algorithm can infer about the
oracle. In particular, once we convert the non-zero entries back to
the primal domain through the unitary $U_p$, i.e., assigning Bernoulli
function values, the likelihood of getting $k$ or more 1s essentially
bounds the probability of successfully solving $\mrsearch$.

\begin{definition}[System state]
For any query algorithm $A$, the state of the entire system in the standard domain, can be described by:
\begin{equation} \label{eq:state_gen}
   \ket{\psi} = \sum_{\vx,\vy,z,D}\alpha_{\vx,\vy,z,D} \ket{\vx,\vy,z} \otimes \ket{D} \, , 
\end{equation}
where $\vx$ and $\vy$ are tuples consisting of $k$ inputs and outputs respectively and $z$ represents the ``workspace'' register of the algorithm.
\end{definition}

 Intuitively, the tuple $\vx$ represents the solution (of size $k$) the adversary outputs for the \mrsearch{} problem and these $k$ inputs may or may not be queried. In this way, after each query, if we were to measure the first registers $\vx$ and $\vy$, we can determine the success probability that the tuple represents a solution for our search problem. 
We will also decompose $\ket{D}$ as:
$$ \ket{D} = \ket{D_\vx} \otimes \ket{D_{\pvx}} \ ,$$
where $D_\vx$ contains the entries that coincide with $\vx$, and $D_{\pvx}$ contains the oracle's state on the rest of the inputs.

Using the definitions of the standard-dual map and of the system state, we can now describe how the system evolves after each query in the dual domain:

\begin{definition}[Query Operator in the Dual Domain]
For the dual domain we define the following query operator:
$$\tilde{O} := ( I \otimes \beru^{\dagger}) \cdot \stdbo \cdot (I \otimes \beru). $$
where $\beru^{\dagger}$ and $\beru$ only act on the last register (containing the strings $D$). 
\end{definition}

\paragraph{Standard-Dual domains view.}
We are now ready to fully describe the Standard-Dual view of the recording technique for Bernoulli functions:
\begin{tiret}
    \item Standard domain:
    \begin{enumerate}
        \item The oracle query is the standard function oracle $\stdbo$;
        \item The initial state of the oracle is the Bernoulli superposition $\sum_f \sqrt{\alpha_f} \ket{f}_F$;
        \item We denote the state of the entire system in the standard domain
          after $t$ queries with $\ket{\psi^t}$.
    \end{enumerate}
    \item Dual domain:
    \begin{enumerate}
        \item The oracle query is $\tilde{O} := ( I \otimes \beru^{\dagger}) \cdot \stdbo \cdot (I \otimes \beru)$;
        \item The initial state of the oracle is the state $\ket{0^M}_F$;
        \item Denote the state of the entire system in the dual domain after $t$ queries with $\ket{\phi^t}$.
    \end{enumerate}
\end{tiret}

It can easily be seen that the two domains are equivalent up to a
final $\beru$ unitary applied on the state of the system corresponding
to the dual domain.

\begin{corollary}[Relation between states in standard-dual domains]
  After any $t$ queries, we have the relation:
\begin{equation}
    \begin{split}
    & \ket{\psi^t} = I \otimes \beru \ket{\phi^t} \\
    & \text{where $\beru$ only acts on the last function register}
    \end{split}
\end{equation}
\end{corollary}

\paragraph*{Evolution of the query algorithm in Dual domain.}
We keep track of the adversary's knowledge in the dual domain, where for the evolution of the system state, we notice the following:
\begin{itemize}
\item After each query, only one position from the function register of each of the states in the superposition of the system state is
  modified, by being rotated using the unitary $U_p$, the rest of the
  positions remain unchanged. \\ 
  Therefore, after $N$ queries there are at most $N$ positions
  containing rotated qubits and the rest of them remain the $\ket{0}$
  state.
\item However, after all $N$ queries, we need to return to the
  standard domain, by applying a final $\beru$ on all qubits of the
  function register and then measure in the computational basis the
  function register. Note that there is a small chance that the
  positions that were never queried (corresponding to $\ket{0}$ state)
  will also collapse to $1$.
\end{itemize}

Once we have established this framework for the recording technique for Bernoulli random functions, we will show how to apply it for our target problem \mrsearch.

%%%%%%%%%%%%%%%%%%%%%%%%%%%%%%%%%%%%%%%%%%%%%%%%%%%%%%%%%%%%%%%%%%%%%%%%%%%%%%%%%%%%%%%%%%%%%%%%%
\subsection{Analysis of  \texorpdfstring{$\mrsearch$}{} via the Bernoulli Recording Technique}     \label{subsec:bound_p_k}
%%%%%%%%%%%%%%%%%%%%%%%%%%%%%%%%%%%%%%%%%%%%%%%%%%%%%%%%%%%%%%%%%%%%%%%%%%%%%%%%%%%%%%%%%%%%%%%%%

To determine the success probability of the adversary to solve the \mrsearch{} problem using $N$ quantum queries, denoted as $p_k^N$, we use the oracle's state as follows:
\begin{newenum}
    \item We look at the vector $\vx$ and at the measurement outcome of the function register (previously described), let us call this measurement $\tilde{f}$. 
    \item Then, we check if $\tilde{f}(x_i) = 1$ for all $x_i \in \vx$.
\end{newenum}
Importantly, the probability of success is defined with respect to the standard domain, hence after the final $\beru$ application.

\paragraph*{Relation between success probability and oracle state.}

In order to relate the success probability $p_k^N$ with the evolution of the oracle we need to define the following set of projectors.
Firstly, we introduce two families of projectors $P$ and $\Pi$, acting on the states of the standard domain  (Eq.~\ref{eq:state_gen}):

\begin{definition}[Projection Family $P$] \label{def:proj_p} For any
  integer $k$, consider the following family of projectors $P$ acting
  in the standard domain:
\begin{tiret}
    \item $P_k$: defined by all basis states $\ket{\vx,\vy,z}\ket{D}$ such that $D$ contains exactly $k$ ones;
    \item $P_{\leq k}$: defined by all basis states $\ket{\vx,\vy,z}\ket{D}$ such that $D$ contains at most $k$ ones, $P_{\leq k} = \sum_{i = 0}^k P_i$;
    \item $P_{\geq k}$: defined by all basis states $\ket{\vx,\vy,z}\ket{D}$ such that $D$ contains at least $k$ ones, $P_{\geq k} = \sum_{i \geq k} P_i$.
\end{tiret}
\medskip
In addition, 
related to $P_k$, we will also define the projectors:
\begin{tiret}
    \item $P_{k}^{0}$ defined by the basis states $\ket{\vx,\vy,z}\ket{D}$ such that $D$ contains exactly $k$ ones, $y = 1$ and $f(x) = 0$.  For this, we can consider that $\ket{x, y}$ is an extra register containing a single input $x$ and a single bit $y$ (different from $\ket{\vx,\vy}$) which is used for the queries to the function oracle \stdbo.
    \item $P_{k}^{1}$ defined by the basis states $\ket{\vx,\vy,z}\ket{D}$ such that $D$ contains exactly $k$ ones, $y = 1$ and $f(x) = 1$.
\end{tiret}
\end{definition}

\begin{definition}[Projection $\Pi$] The projector $\Pi$ is acting in the standard domain and is defined by the basis states $\ket{\vx,\vy,z,D_\vx,D_\pvx}$ such that $D_\vx$ has hamming weight $k$.
\end{definition}

Finally, we introduce the family of projectors $\Xi_i$, which will act on the dual domain (hence before the final application of $\beru$):
\begin{definition}[Projection family $\Xi$] The projectors $\{\Xi_i\}_i$ are acting in the dual domain and are defined by the basis states having exactly $i$ ones qubits in $D_{\vx}$.
\end{definition}

Now, we will see how to employ all these projectors in order to bound the success probability of the adversary to solve the \mrsearch{} problem when having $N$ available quantum queries.

For the success event of finding $k$ preimages of $1$, we must have $k$ ones in $D_{\vx}$ for $\ket{\psi^N}$ - the state in the standard domain. In other words, the adversary's success probability will be bounded by:
\begin{protocol}
\begin{equation}
  p_k^N \leq \norm{\Pi \ket{\psi^N}}^2 = \norm{\Pi \left(\beru
      \ket{\phi^N} \right)}^2 \, .
\end{equation}
\end{protocol}

\paragraph*{Deriving the bound on success probability.}
Next, we will determine the bound on $p_k^N$ by following the next two steps:
\begin{newenum}
    \item Derive a bound on the norm of the projection $\Pi$ using the norms of the projection $\Xi_i$ and subsequently bound this norm with the norms of $P_{\geq i}$.
    \item Bound the norm of $P_{\geq i}$. This will determine the progress after $N$ queries as:
        \begin{equation}
             a_{N, k} := \norm{ P_{\ge k} \ket{\phi^N}}
        \end{equation}
\end{newenum}

\noindent For the first step we show the following result:

\begin{lemma} \label{lemma:bound_pi}
The norm of projection $\Pi$ can be bounded using the progress measure after $N$ queries as follows:
\begin{equation}
    \norm{\Pi \ket{\psi^N}} \leq  \sum_{i = 0}^k (\sqrt{p})^{k - i} (\sqrt{1 - p})^{i} \norm{\Xi_i \ket{\phi^N} }
    \leq  \sum_{i = 0}^k (\sqrt{p})^{k - i} (\sqrt{1 - p})^{i}   \norm{P_{\geq i} \ket{\phi^N}}
\end{equation}
\end{lemma}

\noindent For the second step we will show that:

\begin{lemma} \label{lemma:bound_a_N_k} 
For any $N$ and any $k$, given the
  state $\ket{\phi^N}$ of the dual domain after $N$ queries, we have
  the following bound:
\begin{equation}
    a_{N, k} := \norm{ P_{\ge k} \ket{\phi^N}} \leq 2\sqrt{1 - p} (\sqrt p)^{k} \binom{N}{k}.
\end{equation}
\end{lemma}

In the following sections we will show how to derive these two results and finally we will combine them to obtain the final bound stated in Theorem~\ref{lemma:bound_hard}.

%%%%%%%%%%%%%%%%%%%%%%%%%%%%%%%%%%%%%%%%%%%%%%%%%%%%%%%%%%%%%%%%%
\subsection{Bounding the Success Probability with Progress Measure \label{subsec:proof_1}}
%%%%%%%%%%%%%%%%%%%%%%%%%%%%%%%%%%%%%%%%%%%%%%%%%%%%%%%%%%%%%%%%%

We now show how to bound the success probability using the progress measure, as stated in Lemma~\ref{lemma:bound_pi}.
\begin{proof}[Proof of Lemma~\ref{lemma:bound_pi}]
Firstly, using the definition of $\ket{\psi^N}$ we have:
\begin{equation}
    \norm{\Pi \ket{\psi^N}} = \norm{\Pi \beru \ket{\phi^N}}
\end{equation}
Observe that $\sum_{i = 0}^k \Xi_i = I$, hence we obtain:
\begin{align}
  \norm{\Pi \ket{\psi^N}} = \norm{\Pi \left(\beru  \left(\sum_{i = 0}^k \Xi_i\right) \left(\ket{\phi^N} \right)\right) } 
  \leq \sum_{i = 0}^k \norm{\Pi \left(\beru \Xi_i \left(\ket{\phi^N}\right)\right)}
\end{align}
 
Hence to complete the proof, we need to show that:
\begin{equation}
  \norm{\Pi \left( \beru \Xi_i \left(\ket{\phi^N}\right)\right)} \leq (\sqrt{1 - p})^{i} (\sqrt{p})^{k - i} \norm{\Xi_i \ket{\phi^N}}
\end{equation}

Now let us consider separately the state 
$\ket{\rho_i} := \Xi_i \left(\ket{\phi^N} \right)$ for any quantum state
$\ket{\phi^N}$. Then, from the definition of $\Xi_i$, $\ket{\rho_i}$ has
exactly $i$ ones and $k - i$ zeros in the register $D_{\vx}$.  Let us
denote the set of allowed $D_{\vx}$ (having exactly $i$ ones) by
$\mathcal{D}_{\vx}$, where $|\mathcal{D}_{\vx}| = \binom{k}{i}$; and
denote the set of possible configurations of $D_{\bar{\vx}}$ by
$\mathcal{D}_{\bar{\vx}}$, where
$|\mathcal{D}_{\bar{\vx}}| = 2^{M - k}$. Then we can write:
\begin{equation}
    \ket{\rho_i} = \underset{\substack{ \vx,\vy,z \\ D_{\vx} \in \mathcal{D}_{\vx}, D_{\bar{\vx}} \in \mathcal{D}_{\bar{\vx}}}}{\sum}  \beta_{\vx,\vy,z,D_{\vx}, D_{\bar{\vx}}} \ket{\vx} \ket{\vy} \ket{z} \ket{D_{\vx}} \otimes \ket{D_{\bar{\vx}}}.
  \end{equation}
  
For simplicity we drop the $\ket{\vy} \ket{z}$ register as they will
not be affected, and instead denote $\ket{\rho}$ as:
\begin{equation}
    \ket{\rho_i} = \underset{\substack{ \vx,\vy,z \\ D_{\vx} \in \mathcal{D}_{\vx}, D_{\bar{\vx}} \in \mathcal{D}_{\bar{\vx}}}}{\sum} \beta_{\vx,\vy,z,D_{\vx}, D_{\bar{\vx}}} \ket{\vx} \ket{D_{\vx}} \otimes \ket{D_{\bar{\vx}}}.    
\end{equation}

Now let us consider applying $\Pi$ to a single state from this
superposition, namely
$\ket{\vx}\ket{D_{\vx}} \otimes \ket{D_{\bar{\vx}}}$:
\begin{eqnarray}\label{eq:bs1}
    \norm{ \Pi \ket{\vx} (\beru \ket{D_{\vx}} \otimes \ket{D_{\bar{\vx}}})}  &=& || \Pi \ket{\vx} (\sqrt{p}\ket{0} - \sqrt{1 - p}\ket{1})^{\otimes i} \otimes (\sqrt{1 - p}\ket{0} + \nonumber\\
    &+& \sqrt{p}\ket{1})^{\otimes k - i} \otimes U_p^{\otimes M - k} \ket{D_{\bar{\vx}}} || \nonumber \\
    &=& (\sqrt{1 - p})^i (\sqrt{p})^{k - i} \norm{\ket{\vx} \otimes \ket{1_{D_{\vx}}} \otimes U_p^{\otimes M - k} \ket{D_{\bar{\vx}}}} \nonumber\\
    &=& (\sqrt{1 - p})^i (\sqrt{p})^{k - i} 
\end{eqnarray}

Therefore, we have that:
\begin{align} \label{eq:bs2}
  \norm{\Pi \beru \ket{\rho_i}} &= (\sqrt{1 -
                                p})^i (\sqrt{p})^{k - i} \norm{\underset{\substack{ \vx,\vy,z
  \\
  D_{\vx} \in \mathcal{D}_{\vx}, D_{\bar{\vx}} \in \mathcal{D}_{\bar{\vx}}}}{\sum} \beta_{\vx,\vy,z,D_{\vx}, D_{\bar{\vx}}} \ket{\vx}  \otimes \ket{1_{D_{\vx}}} \otimes \left(U_p^{\otimes M - k} \ket{D_{\bar{\vx}}} \right)} \\
                              & = (\sqrt{1 - p})^i (\sqrt{p})^{k - i} \cdot
                                \norm{\underset{\substack{ \vx,\vy,z \\ D_{\vx} \in
  \mathcal{D}_{\vx}, D_{\bar{\vx}} \in \mathcal{D}_{\bar{\vx}}}}{\sum}
  \beta_{\vx,\vy,z,D_{\vx}, D_{\bar{\vx}}} \ket{\vx} \otimes
  \ket{1_{D_{\vx}}} \otimes \ket{D_{\bar{\vx}}}} \\
                              & =  (\sqrt{1 - p})^i (\sqrt{p})^{k - i} \cdot \sqrt{\underset{\substack{ \vx,\vy,z \\ D_{\vx} \in \mathcal{D}_{\vx}, D_{\bar{\vx}} \in \mathcal{D}_{\bar{\vx}}}}{\sum} |\beta_{\vx,\vy,z,D_{\vx}, D_{\bar{\vx}}}|^2} \\
                              &= (\sqrt{1 - p})^i (\sqrt{p})^{k - i} \cdot \norm{\ket{\rho_i}}
\end{align}

\end{proof}

%%%%%%%%%%%%%%%%%%%%%%%%%%%%%%%%%%%%%%%%%%%%%%%%%%%%%%%%%%%%%%%%%
\subsection{Bounding the Progress Measure}      \label{bound_a_N_k}
%%%%%%%%%%%%%%%%%%%%%%%%%%%%%%%%%%%%%%%%%%%%%%%%%%%%%%%%%%%%%%%%%

We now show how to bound the progress measure $a_{N, k} := \norm{ P_{\ge k} \ket{\phi^N}}$, as stated in Lemma~\ref{lemma:bound_a_N_k}.

\begin{proof}[Proof of Lemma~\ref{lemma:bound_a_N_k}]

First, we want to define a recurring relation for the progress measure sequence. Specifically, for any $t$, we have:
\begin{equation}\label{eq:rec_prog_measure}
    \begin{split}
    a_{t + 1, k} &= \norm{ P_{\ge k} \tilde{O} \ket{\phi^t} } = \norm{ P_{\ge k} \tilde{O} (P_{\ge k} + P_{k - 1} + P_{\leq k - 2}) \ket{\phi^t} } \\
    & \leq \norm{ P_{\ge k} \tilde{O} P_{\ge k} \ket{\phi^t}} + \norm{ P_{\ge k} \tilde{O} P_{k - 1} \ket{\phi^t}} \\
    & \leq a_{t, k} + \norm{ P_{\ge k} \tilde{O} P_{k - 1} \ket{\phi^t}}
    \end{split}
\end{equation}
Now let us consider separately the term $\norm{ P_{\ge k} \tilde{O} P_{k - 1} \ket{\phi^t}}$. A basis state $\ket{\vx,\vy,z}\ket{D}$ contributes to $a_{t+1,k}$ if $D$ has exactly $k - 1$ ones, but also if $y = 1$ and $f(x) \neq 1$. \\

What would actually contribute to $a_{t+1,k}$ is the state
$P_{\ge k} \tilde{O} P_{k - 1}^0 \ket{\phi^t}$, hence we have:
\begin{equation} \label{eq:b1}
    a_{t + 1, k} \leq a_{t, k} + \norm{P_{\ge k} \tilde{O} P_{k - 1}^{0} \ket{\phi^t}}
\end{equation}
Now, to bound the second term, for any basis state $\ket{\vx,\vy,z}\ket{D}$ in the support of $P_{k - 1}^{0}$,  we will make use of the following observation.

Using the definition of $\tilde{O}$, we notice that each query $\ket{x, y}$ in the dual domain has the following effect on the zero entries of the $D$ strings (oracle state) - when $f(x) = 0$ :
\begin{equation}\label{eq:up_squared}
    \begin{split}
 U_p^{\dagger} \cdot \stdbo \cdot U_p \ket{0} &= U_p^{\dagger} [(-1)^{y \cdot 0} \sqrt{1 - p} \ket{0} + (-1)^{y \cdot 1} \sqrt{p} \ket{1}] = {U_p^{\dagger}}^{2y} \ket{0} \\
 & = \begin{cases}
      \ket{0}  \text{, } &\text{if } y = 0 \\
      (1 - 2p) \ket{0} - 2 \sqrt{p(1 - p)} \ket{1}  \text{, } &\text{if } y = 1 \\ 
     \end{cases}
    \end{split}
\end{equation}

This leads us to:
\begin{equation}\label{eq:b2}
 P_{\ge k} \tilde{O} \ket{\vx,\vy,z} \otimes \ket{D} = -2\sqrt{p(1-p)} \ket{\vx,\vy,z} \otimes D_{\vx - \{x\}} \otimes \ket{1}_x\otimes \ket{D_{\bar{\vx}}})   
\end{equation}
where $(x, y)$ is the query to the function oracle \stdbo, as in Definition~\ref{def:proj_p}. \\
As a result, we have:
\begin{equation}
    \norm{P_{\ge k} \tilde{O} P_{k - 1}^{0} \ket{\phi^t}} = 2\sqrt{p(1-p)} || P_{k - 1}^{0} \ket{\phi^t} ||
\end{equation}
Therefore, we obtain:
\begin{equation}
    a_{t + 1, k} \leq a_{t, k} +  2\sqrt{p(1-p)} || P_{k - 1}^{0} \ket{\phi^t} ||
\end{equation}

Applying this relation iteratively for any $t \in \{N - 1, ..., 0\}$, leads us to the following bound on the progress measure after $N$ queries:
\begin{equation}
    a_{N, k} \leq 2 \sqrt{p(1 - p)} \sum_{i = 0}^{N - 1} \norm{P_{k - 1}^{0} \ket{\phi^i}}
\end{equation}

By using Lemma~\ref{lemma:betaalli}, we obtain the bound on $a_{N, k}$:
\begin{equation}
    a_{N, k} \leq 2\sqrt{1 - p} (\sqrt p)^{k} \binom{N}{k}.
\end{equation}

\end{proof}

\begin{lemma} Let $\ket{\phi^i}$ be the state before the $i + 1$'th query and $\beta_{k-1}^i : = \norm{P_{k-1}^{0}\ket{\phi^i}}$.
Then $S_{k, N} := \sum_{i = 0}^{N-1} \beta_{k-1}^i \leq \binom{N}{k}(\sqrt p)^{k-1}$.
\label{lemma:betaalli}
\end{lemma}

\begin{proof}

We analyze $\beta_{k-1}^i$ case by case:
\begin{itemize}
    \item If $i < k - 1$, then $\beta_{k - 1}^i = 0$. This is because unless at least $k - 1$ queries have been made, the number of non-zero entries in $D$ is below $k - 1$ and falls out of the support of $P_{k-1}^{0}$.
    \item If $i = k - 1$, then $\beta_{k - 1}^i \leq (\sqrt{p})^{k - 1}$. By applying $k - 1$ queries the magnitude of all resulting in $\ket{1}$ is $\sqrt{p}^{k - 1}$.
    \item If $i \geq k$, then $\beta_{k - 1}^i \leq \binom{i}{{k-1}}
      (\sqrt{p})^{k-1}(\sqrt{1-p})^{i - k + 1}$. This follows from a combinatorial argument, reminiscent of binomial coefficients. Think of taking a random walk on a line for $i$ steps, we need to move forward $k - 1$ times. There are $\binom{i}{k - 1}$ ``routes'', and each route contributes a weight of at most $(\sqrt p)^{k-1}(\sqrt{1-p})^{i - k + 1}$.
\end{itemize}
Hence we have:
\begin{equation}
    S_{k, N} := \sum_{i = 0}^{N-1} \beta_{k-1}^i \leq (\sqrt{p})^{k - 1} \sum_{i = k - 1}^{N - 1}\binom{i}{k-1} = (\sqrt{p})^{k - 1} \binom{N}{k}.
\end{equation}
where for the last equality we used the identity: $\sum_{n = t}^{m} \binom{n}{t} = \binom{m+1}{t+1}$.
\end{proof}

%%%%%%%%%%%%%%%%%%%%%%%%%%%%%%%%%%%%%%%%%%%%%%%%%%%%%%%%%%%%%%%%%
\subsection{Putting Everything Together}
%%%%%%%%%%%%%%%%%%%%%%%%%%%%%%%%%%%%%%%%%%%%%%%%%%%%%%%%%%%%%%%%%

We now can combine the results to obtain the bound on $p_k^N$:

\begin{theorem}
  The success probability $p_k^N$ of solving $\mrsearch$ satisfies
  \begin{equation}
    p_k^N \leq 4(1-p) p^k \left( \sum_{i = 0}^k (\sqrt{1 - p})^i \cdot
      \binom{N}{i} \right)^2 \, .
\end{equation}
\end{theorem}

\begin{proof}
First by Lemma~\ref{lemma:bound_pi} we have 
\begin{align}
  p_k^N \leq \norm{\Pi\ket{\psi^N}}^2 \leq \sum_{i = 0}^k (\sqrt{p})^{k
  - i} (\sqrt{1 - p})^i \norm{ P_{\geq i} \ket{\phi^N}} \, .
\end{align}

We know from Lemma~\ref{lemma:bound_a_N_k} that
\begin{equation}
    a_{N, i} := \norm{ P_{\ge i} \ket{\phi^N}} \leq 2 \sqrt{1 - p} (\sqrt{p})^i \binom{N}{i} 
\end{equation}

Hence we have:
\begin{equation} \label{eq:p_k_N_a_k_N} p_k^N \leq \left( \sum_{i =
      0}^k (\sqrt{p})^{k - i} (\sqrt{1 - p})^i a_{N, i} \right)^2 \le
  4(1-p) p^k \left( \sum_{i = 0}^k (\sqrt{1 - p})^i \cdot \binom{N}{i}
  \right)^2 \, .
\end{equation}

\end{proof}

%%%%%%%%%%%%%%%%%%%%%%%%%%%%%

% 4. QUANTUM BACKBONE

%%%%%%%%%%%%%%%%%%%%%%%%%%%%%

\section{The Bitcoin Backbone against Quantum Adversaries} \label{sec:bitcoin}

%%%%%%%%%%%%%%%%%%%%%%%%%%%%%%%%%%%%%%%%%%%%%%%%%%%%%%%%%%%%%%%%%%%%%%%%%%%%%%%%%%%%%%%%
\subsection{The Bitcoin Backbone Protocol: Model and Definitions \label{sec:model_defs}}
%%%%%%%%%%%%%%%%%%%%%%%%%%%%%%%%%%%%%%%%%%%%%%%%%%%%%%%%%%%%%%%%%%%%%%%%%%%%%%%%%%%%%%%%

We will analyze our post-quantum version of the Bitcoin backbone protocol, 
where the honest parties are classical but the adversary is not,
in the network model considered in~\cite{GKL15}, namely, a synchronous communication network which is based  on  Canetti's formulation  of ``real world'' execution for multi-party cryptographic protocols~\cite{Canetti2000,DBLP:conf/focs/Canetti01}).

In such a network model, the protocol execution proceeds in {\em rounds}; in each round,  parties receive messages from the network through their communication interface (e.g., input tape or register), perform some computation, and send messages to the network, which are delivered at the beginning of the next round. 
The inputs to the computation performed by the parties are provided by an environment program denoted by $\mathcal{Z}$. For further details about the synchronous network model, refer to~\cite{Canetti2000}. 

The execution is assumed to have a polynomial (in the security parameter) time bound. 
The actual message delivery is provided by a ``diffusion'' mechanism that is guaranteed to deliver all messages, without however preserving their order and allowing the adversary to arbitrarily inject its own messages. Importantly, the parties are not guaranteed to have the same view of the messages delivered in each round, except for the fact that all the messages sent by the honest parties in 
the previous round are delivered. Furthermore, 
and in line with the cryptographic protocols literature, we assume the presence of a
single adversary (which may in turn ``corrupt'' many parties), albeit 
equipped with quantum computing power.
The adversary is allowed to ``spoof'' messages by
changing the source information in them
(i.e., communication is not authenticated). 

\paragraph{The Bitcoin backbone protocol.} First, we introduce some blockchain notation, following~\cite{GKL15}. A  \emph{block} is any triple of the form $B=\langle s, x, ctr\rangle$  where $s\in \{0,1\}^\kappa, x \in \{0,1\}^*, ctr\in\mathbb{N}$ are such that satisfy
predicate $\mathsf{validblock}^D_q(B)$ defined as:
\begin{equation}
\label{eq:validblock}
    ( H( ctr, G(s, x))  < T ) \land (ctr\leq q), 
\end{equation}

\noindent where $H, G$ are cryptographic hash functions (e.g., SHA-256) modelled as random oracles.
The parameter  $T \in \mathbb{N}$ is also called the block's {\em difficulty level}. We then define $p = T / 2^\kappa$ to be the probability that a single classical query solves a \pow.
The parameter $q \in\mathbb{N}$ is a bound that in the Bitcoin implementation determines the size
of the register $ctr$; in our treatment we allow this to be arbitrary, and use it to denote  the maximum allowed number of hash queries performed by the (classical) parties in a round.

A {\em blockchain}, or simply a {\em chain} is a sequence of \emph{blocks}. 
The rightmost block is the {\em head} of the chain, denoted $\mathrm{head}(\chain)$. 
Note that the empty string $\varepsilon$ is also a chain; by convention we set $\mathrm{head}(\varepsilon) = \varepsilon$. 
A chain $\chain$  with $\mathrm{head}(\chain) = \langle s',x',ctr'\rangle$ can be extended to a longer chain by appending 
a valid block $B  = \langle s, x, ctr  \rangle$ that satisfies  $s = H( ctr', G(s',x') )$. 
In case $\chain=\varepsilon$, by convention any valid block of the form $\langle s,x,ctr\rangle$ may extend it. 
In either case we have an extended chain $\chain_\mathsf{new} = \chain B$ that satisfies $\mathrm{head}(\chain_\mathsf{new}) = B$. 
Consider a chain $\chain$ of length $m$ (written as $len(\chain) = m$) and any non-negative integer $k$. 
We denote by $\chain^{\lceil k}$  the chain resulting from the ``pruning'' of the $k$ rightmost blocks. 
Note that for $k\ge len(\chain)$, $\chain^{\lceil k}=\varepsilon$. 
If $\chain_1$ is a prefix of $\chain_2$ we write $\chain_1 \preceq \chain_2$. 

The Bitcoin backbone protocol is executed by an arbitrary number of parties over an unauthenticated network, as described above. It is assumed in~\cite{GKL15} that the number of parties running the protocol is fixed, however, parties need not be aware of this number when they execute the protocol. In our analysis we will have $n$ honest parties and a single quantum adversary. 
Also as mentioned above, communication over the network is achieved by utilizing a send-to-all {\sc Diffuse} functionality that is available to all parties (and may be abused by the adversary in the sense  of delivering different messages to different parties).  

Each party maintains a blockchain, as defined above, starting from the empty chain and mining a block that contains
the value $s=0$ (by convention this is the ``genesis block'').
If in a given round, an honest party is successful in generating a \pow{} (i.e. satisfying conjunction~\ref{eq:validblock}), it diffuses it to the network.
At each round, each party chooses the longest chain amongst the one he received, and tries to extend it by computing (mining) another block. In such a process, each party's chain may be different, but under certain well-defined conditions, it is shown in~\cite{GKL15} that the chains of honest parties will share a large common prefix (see below). 

In the backbone protocol, the type of values that parties try to insert in the chain is intentionally left unspecified, as well as the type of chain validation they perform (beyond checking for its structural properties with respect to the hash functions $G(\cdot),H(\cdot)$), and the way they interpret the chain. Instead, these actions are abstracted by the external functions $V(\cdot)$ (the {\em content validation predicate}), $I(\cdot)$ (the {\em input contribution function}), and $R(\cdot)$ (the {\em chain reading function}), which are specified by the  application that runs ``on top'' of the backbone protocol (e.g., a transaction ledger).

\paragraph{Basic security properties of the blockchain.} It is shown
in~\cite{GKL15} that the blockchain data structure built by the
Bitcoin backbone protocol satisfies a number of basic properties. At a
high level, the first property, called {\em common prefix}, has to do
with the existence, as well as persistence in time, of a common prefix
of blocks among the chains of honest parties. The second, called
\emph{chain quality}, stipulates the proportion of honest blocks in
any portion of some honest party’s chain.

\begin{definition}[Common Prefix]
The {\em common prefix} property with parameter $k \in \mathbb{N}$, 
states that for any pair of honest players $P_1, P_2$ adopting chains $\chain_1, \chain_2$ at rounds $r_1 \leq r_2$, it holds that $\chain_1^{\lceil k} \preceq \chain_2$ (the chain resulting from pruning the $k$ rightmost blocks of $\chain_1$ is a prefix of $\chain_2$).
\end{definition}

\begin{definition}[Chain Quality] \label{def:chain_quality}
The {\em chain quality} property with parameters $\mu \in \mathbb{R}$ and $l \in \mathbb{N}$, states that for any honest party $P$ with
chain $\chain$, it holds that for any $l$ consecutive blocks of $\chain$, the ratio of blocks created by honest players is at least $\mu$.
\end{definition}

\paragraph{Parameters and Random Variables.}
Next, we recall some important notions in the Bitcoin backbone protocol setting.
\begin{tiret}
\item $n$ denotes the number of honest parties;
\item $q$ represents the number of classical queries of each honest party per round;
\item $Q$ denotes the number of adversarial quantum queries per round;
\item $f$ is the probability that at least one honest party generates a \pow{} (i.e. satisfy conjunction~\ref{eq:validblock}) in a round;
\item $\epsilon$ will be used for the concentration quality of random variables while $\kappa$ for the security parameter;
\item $k$ denotes the number of blocks for common prefix property and $\mu$ denotes the chain quality parameter;
\item $s$ refers to the total number of rounds;
\item $p$ is the probability of success of generating a \pow{} (conjunction~\ref{eq:validblock}) using a single classical query;
\item $X(s)$ and $Y(s)$ represent within $s$ rounds, the numbers of
rounds at least one honest player and \emph{exactly} one honest player
solves a \pow{} respectively. 
\item $Z(s)$ represents the number of \pows{} solved
by an adversary in $s$ consecutive rounds.
\item $f$ denotes the
probability that at least one honest player generates a \pow{} (conjunction~\ref{eq:validblock}) in a single
round (e.g., in the Bitcoin system, $f$ is about $2-3\%$).
\item In condition (c), an \textit{insertion} denotes the event that given a
chain $\chain$ with two consecutive blocks $B$ and $B'$, a block $B^*$
created after $B'$ is such that $B$, $B^*$, $B'$ form three
consecutive blocks of a valid chain. 
\item A \textit{copy} occurs if the
same block exists in two different positions. 
\item A {\em prediction}
occurs when a block extends one which was computed at a later round.
\end{tiret}

\subsection{The \chainpow{} Problem}

In the Bitcoin backbone protocol, an adversary aims to produce a chain of blocks that is longer than the honest chain. We formalize it as the \emph{\chainpow{}} search problem below.

\begin{protocol}
\textsc{Problem $\Pi_G$: \chainpow} \\
\noindent\textbf{Given}: $N$, $x_0\in X$ and $h_0,\ldots, h_{N-1}$ as (quantum) random oracles, where each $h_i: X\times Y \to X$ is independently sampled. \\ 
\noindent\textbf{Goal}: Using $N$ total queries find a sequence  $y_0,\ldots, y_{k-1}$ such that $x_{i+1}: = h_{i}(x_{i},y_{i})$ and $x_{i+1} \leq T$ $\forall \ i \in \{0, \cdots, k-1 \}$ such that the length of the sequence $k\leq N$ is the maximum that can be achieved. 
$T$ is a fixed positive integer.
\end{protocol}

Note that the output of this problem is the maximum length $k$ and the
corresponding sequence $(y_0, ..., y_{k-1})$. For ease of notations,
we will omit $h_0,\ldots,h_{N-1}$ from the input of $\Pi_G$. We call
any pair $(x,y)$ s.t.  $h(x,y) \leq T$ a \pow.

A possible approach is to prove a composition theorem for the query
complexity of such a search problem. This however appears beyond the
scope of existing results and techniques in quantum query
complexity. Instead, we show its hardness by reducing to the
$\mrsearch$ problem we analyzed before. 

\begin{theorem}[Main Theorem] \label{thm:main_result}
For any quantum adversary $\mathcal{A}$ having $N$ quantum queries, the probability that $\mathcal{A}$ solves the \chainpow{} problem, by outputting a solution of size at least $k$ is at most: 
\begin{equation}
 P(N, k) \leq \exp \left(-2k \cdot \ln \left( \frac{k}{e(N + k)} \cdot \frac{1}{\sqrt{p}}\right) \right)   
\end{equation}
where $p := \frac{T}{2^{\kappa}}$ is the probability of success of a single query to the random oracle.
\end{theorem}

%%%%%%%%%%%%%%%%%%%%%%%%%%%%%%%%%%%%%%%%%%%%%%%%%%%%%%%%%%%%%%%%%%%%
\subsection{Hardness of \chainpow}   \label{sec:cpows}
%%%%%%%%%%%%%%%%%%%%%%%%%%%%%%%%%%%%%%%%%%%%%%%%%%%%%%%%%%%%%%%%%%%%

We relate $\chainpow$ to $\mrsearch$ formally below. 

\begin{lemma} \label{lemma:bag_chain_prob}
If for any quantum adversary $\mathcal{A}$ having $N$ quantum queries the probability that $\mathcal{A}$ obtains a solution of size $k$ for \mrsearch{} is at most $p_1(N, k)$, then for any quantum adversary $\mathcal{A}'$ having $N$ quantum queries, the probability that $\mathcal{A}'$ obtains a solution of size $k$ for \chainpow{} is at most $p_2(N, k)$, satisfying: $p_2(N, k) \leq p_1(N + k, k)$.
\end{lemma}

\begin{proof}

We first state \mrsearch{} in an equivalent form, which we call \bagpow{}. That is, we consider $N$ functions $h_0,\ldots, h_{N-1}$ as oracles, where each $h_i: X\times Y \to X$ is independently sampled, and the goal is to find a set of pairs $\{(x_{i_1},y_{i_1})_1,$ $\ldots, (x_{i_k},y_{i_k})_k\}$ so that $ h_{i_l}(x_{i_l},y_{i_l}) \leq T$, for all $l\in \{1,\ldots, k\}$. This is clearly equivalent to \mrsearch, by considering the random boolean function $f : [N] \times X \times Y \rightarrow \{0, 1\}$: $f(i, x, y) = 1$ if and only if $h_i(x, y) \leq T$. 

Now suppose for the sake of contradiction that there exists a quantum adversary $\mathcal{A}'$ equipped with $N$ queries such that:
\[ p_2 := \Pr[\mathcal{A}'(N) \text{ gives a solution of size } k \text{ for \chainpow{}}] > p_1(N + k, k) \, .\]
We construct an algorithm $\mathcal{A}''$ that using $N + k$ queries obtains a solution of size $k$ for the \bagpow{} problem, and hence a solution also for the \mrsearch{}
problem, with the same probability $p_2$.

$\mathcal{A}''$ first samples at random an element $x_0$ (from the domain of any oracle $h_i$). Then, it will run $\mathcal{A}'$ on input $(N, x_0)$, in order to obtain the sequence $y_0, \ldots, y_{k-1}$ with probability $p_2$ spending $N$ queries. Next, starting from $y_0$, $\mathcal{A}''$ computes $x_i = h_{i-1}(x_{i - 1}, y_{i - 1})$ for all $i \in \{1, \cdots, k - 1\}$ (which is guaranteed to hold from the setting of \chainpow) spending one extra query per element of the sequence, i.e. extra $k$ queries. Then the algorithm $\mathcal{A}''$ outputs the pairs $(x_0,y_0), (x_1, y_1), \ldots\ (x_{k - 1}, y_{k - 1})$ having used $N + k$ queries and will succeed with probability $p_2 > p_1(N + k, k)$.

But from the initial assumption, the probability that $\mathcal{A}''$ using $N + k$ queries obtains a solution of size $k$ for the \mrsearch{} is at most $p_1(N + k, k)$. 

\end{proof}

By combining  Lemma~\ref{lemma:bag_chain_prob} together with the bound on the success probability for the \mrsearch{} problem from Theorem~\ref{lemma:bound_hard}, we can obtain a bound on the success probability for the \chainpow{} problem:

\begin{lemma}\label{lemma:chain_pows}
Any quantum adversary having $N$ queries to the QRO, can obtain a chain of \pows{} of length $k$, with probability at most:
\begin{equation}
P(N, k) := \frac{2(1 - p)}{\pi k} \left[ \frac{(N + k)e\sqrt{p}}{k} \right]^{2k}   
\end{equation}
where $p$ is the probability of a successful \pow{} with a single query.
\end{lemma}

%%%%%%%%%%%%%%%%%%%%%%%%%%%%%%%%%%%%%%%%%%%%%%%%%%%%%%%%%%%%%%%%%%%%%%%%%%%%%%%%%%%%%%
\subsection{Post-Quantum Security of the Bitcoin Backbone Protocol}  \label{sec:pqbca}
%%%%%%%%%%%%%%%%%%%%%%%%%%%%%%%%%%%%%%%%%%%%%%%%%%%%%%%%%%%%%%%%%%%%%%%%%%%%%%%%%%%%%%%

Now that we have established the hardness of the $\chainpow$ problem, we
will use it to identify the conditions under which the central
properties of blockchains, \emph{common prefix} and
\emph{chain quality}, can be satisfied in the presence of quantum
adversaries. To do so, it is helpful to take a closer look at the
classical analysis~\cite{GKL15}, which proceeds in three steps.

\begin{newenum}
\item A notion of \emph{typical execution} of a blockchain is defined
  in the presence of honest players and adversaries. This is a
  critical vehicle that modularizes the analysis.
\item It is proven that under a typical execution, the desired
  properties will follow for suitable parameters.
\item An \emph{honest majority} condition is identified under which a typical execution occurs with high probability. 
\end{newenum}

We would like to ``lift'' this analytical framework to the quantum setting
with as few changes as possible. However, we run into a roadblock
immediately since the definition of a typical execution dose not make
sense any more. To see this, let us recall the classical definition.

\begin{definition}[Classical typical execution~\cite{GKL15}]\label{def:classical-typical-execution} An execution is $(\epsilon, s)$-typical (or just typical), for $\epsilon \in (0, 1)$, if for any set $S$ of at least $s$ consecutive rounds, we have:
  \begin{enumerate}[label=(\alph*)]
  \item $(1 - \epsilon)fs < X(s) < (1 + \epsilon)fs 
\ \text{ and } \ (1 - \epsilon)\mathbb{E}[Y(s)] < Y(s)$.
  \item $Z(s) < \mathbb{E}[Z(s)] + \epsilon \mathbb{E}[X(s)]$.
  \item No insertions, no copies, and no predictions
    occurred.
  \end{enumerate}
  \label{def:texec}  
 
\end{definition}

For completeness, let us first revisit now four main results of \cite{GKL15} which would help us extract the necessary conditions on the number of adversarial \pows{} required to satisfy the security of the backbone protocol.

\begin{lemma}[\cite{GKL15}] \label{lemma:conc_x_y} 
For any $s \geq {2}/{f}$ rounds, we have that with probability $1 - e^{-\Omega(\epsilon^2sf)}$, the following hold:
\begin{equation}
(1 - \epsilon)fs < X(s) < (1 + \epsilon)fs 
\ ; \ (1 - \epsilon)\mathbb{E}[Y(s)] < Y(s) \ ; \ (1 - \epsilon)f(1 - f)s < Y(s). \nonumber
\end{equation}
\end{lemma}

\begin{lemma}[\cite{GKL15}] \label{lemma_first_condition_lemma_5}
Assuming that $X(s) + Z(s) < 2fs$, then in a typical execution, any $k \geq 2fs$ consecutive blocks of a chain have been computed in more than $k/2f$ consecutive rounds.
\end{lemma}

\begin{lemma}[Common Prefix Lemma (\cite{GKL15})] \label{lemma_second_condition_lemma_6}
Assuming that $Z(s) < Y(s)$ and that at round $r$ of a typical execution an honest party has a chain $\mathcal{C}_1$, while a chain $\mathcal{C}_2$ of length at least $len(\mathcal{C}_1)$ is adopted by an honest party. Then $\mathcal{C}_1^{\lceil{k}} \preceq \mathcal{C}_2$ and $\mathcal{C}_2^{\lceil{k}} \preceq \mathcal{C}_1$ for $k \geq 2fs$ and $s \geq \frac{2}{f}$. 
\end{lemma}

\begin{lemma}[Chain Quality (\cite{GKL15})] \label{lemma:chain_quality_cond_lemma_7}
In a typical execution, the chain quality property (Definition~\ref{def:chain_quality}) holds with parameters $l \geq 2sf$ and ration of honest blocks $\mu$ such that the following condition holds: $Z(s) < (1 - \mu)X(s)$.
\end{lemma}

\paragraph{Post-quantum typical executions.}
 As it turns out, the definition of classical typical executions (Def.~\ref{def:classical-typical-execution}) 
is not adequate for our post-quantum analysis, for the following main reason.
Note that in the classical random oracle model, in terms of solving a
\pow, everyone including malicious players runs the same
procedure. Hence, there is a \emph{universal} well-defined
distribution on $Z(s)$, regardless what an adversary does otherwise,
based on which we can discuss its expectation.
When we consider a quantum adversary, however, it is unclear how to
formulate an appropriate ``expectation'' that is independent of the
specific adversary's quantum strategy.

To address this, we observe that some alternative quantitative
characterizations of a typical execution in~\cite{GKL15} actually
admit simple counterparts in the quantum setting. Hence we adopt these
as our notion of typical executions below. As a result, we can lift
the classical analysis in step 2 almost verbatim to show that
common prefix and chain quality hold under this new definition of
typical execution. 

Finally, we will use the hardness of $\chainpow$ to
derive a quantum analogue of a honest-majority condition, under which
typical execution occurs with overwhelming probability.

\begin{definition}[Post-quantum typical execution] \label{def:pq_typical_exec}
  An execution is post-quantum $(\varepsilon,s)$-typical (or just post-quantum typical), for
  $\varepsilon \in (0,1)$ and $sf\ge 2$, if for any set $S$ of at least
  $s$ consecutive rounds, the following hold:

  \begin{enumerate}[label=(\alph*)]
  \item $(1 - \epsilon)fs < X(s) < (1 + \epsilon)fs 
\ \text{ and } \ (1 - \epsilon)\mathbb{E}[Y(s)] < Y(s)$.
    \item  $Z(s) < (1 - \epsilon)f(1 - f)s$.
  \item No insertions, no copies, and no predictions occurred. 
  \end{enumerate}
  \label{def:pq-texec}  
\end{definition}

Note that Conditions (a) and (c) remain unchanged, and in particular,
(a) concerns honest parties only. 
As in the classical setting, under our new definition of a typical
execution, the desired properties of a blockchain follow easily.
The second condition follows from the derived bounds on the hardness of the \chainpow{} problem.

\begin{lemma}
\label{thm:common_prefix_chain_quality_condition}
In a post-quantum typical execution in the presence of any quantum
adversary,
\begin{tiret}
\item The common prefix property of the Bitcoin backbone protocol holds with parameter $k \geq 2sf$, for any $s \geq \frac{2}{f}$ consecutive rounds; 
\item The chain quality property holds with parameter $l \geq 2sf$ and ratio of honest blocks $\mu$ with $\mu = f$.
\end{tiret}
\end{lemma}

\begin{proof}

Using Lemma~\ref{lemma_first_condition_lemma_5} we know that we can ensure that any $k \geq 2fs$ consecutive blocks of a chain have been computed in $s \geq \frac{k}{2f}$ consecutive rounds, as long as we can impose the following condition:
for any quantum adversary $\mathcal{A}$ and for any $s \geq \frac{2}{f}$, we have: $ X(s) + Z(s) < 2fs$. \\
Then, in a typical execution (Def.~\ref{def:pq_typical_exec}) we have that $X(s) < (1 + \epsilon) fs $ which implies that we must have $\frac{Z(s)}{s} < (1 - \epsilon)f$.

Secondly, the condition between $Z(s)$ and $Y(s)$ comes from Lemma~\ref{lemma_second_condition_lemma_6}, which then implies the common prefix property.
To apply Lemma~\ref{lemma_second_condition_lemma_6}, what we must guarantee is that
for any quantum adversary $\mathcal{A}$ and for $s \geq \frac{2}{f}$, we have: $Z(s) < Y(s)$.
Therefore, in order to prove that the common prefix property holds with parameter $k \geq 2sf$, it is sufficient to impose on the quantum adversary the following two conditions for any $s \geq \frac{2}{f}$ consecutive rounds: 
\begin{equation}
X(s) + Z(s) < 2fs \ \ ; \ \ 
Z(s) < Y(s)
\end{equation}
Using the bounds on the honest players variables $X(s)$ and $Y(s)$ from Lemma~\ref{lemma:conc_x_y}, the sufficient conditions become:
\begin{equation}
\frac{Z(s)}{s} < (1 - \epsilon)f \ \ ; \ \
 \frac{Z(s)}{s} < (1 - \epsilon)f(1 - f)
\end{equation}
Given that $ \min\{(1 - \epsilon)f(1 - f), (1 - \epsilon)f\} = (1 - \epsilon)f(1 - f)$, this leads to:
\begin{equation}
\label{eq:restriction_Q}
    \frac{Z(s)}{s} < (1 - \epsilon)f(1 - f)
\end{equation}
which is exactly the second condition of the post-quantum typical execution (Def.~\ref{def:pq_typical_exec}). \\
For the chain quality property, the statement follows directly from the proof of chain quality in Lemma~\ref{lemma:chain_quality_cond_lemma_7}.

\end{proof}

Finally we are just left to find out an appropriate quantum analogue of the honest-majority condition, and show that it will ensure a post-quantum typical execution occurs with high probability.

\begin{definition}[Post-quantum honest majority]
  We say that the post-quantum honest majority condition\footnote{See the proof of Theorem~\ref{thm:texec} for intuition on how this expression arises.} holds
  if: 
\begin{equation}
    Q \leq \frac{(1 - \epsilon)f(1 - f)}{(1 + \epsilon)e \sqrt{p}}
\end{equation}

\noindent where $Q$ denotes the total number of quantum queries performed by the adversary per round.
\label{def:pq_honest_majority}
\end{definition}

\begin{lemma} 
Under the post-quantum honest majority condition  (Def.~\ref{def:pq_honest_majority}), the probability of a post-quantum typical execution is:
\begin{equation}
    P_{\text{q}} = 1 - e^{-\Omega((1 - \epsilon) f (1 - f)s)}.
\end{equation}
  \label{thm:texec}
\end{lemma}

\begin{proof}

We examine the bound we obtained in Lemma~\ref{lemma:chain_pows} for the \chainpow{} problem. A simplified form of the bound on this probability success (derived by using $\frac{2(1 - p)}{\pi k} < 1$), can be expressed as:
\begin{equation}
 P(N, k) \leq \exp \left(-2k \cdot \ln \left( \frac{k}{e(sQ + k)} \cdot \frac{1}{\sqrt{p}}\right) \right)   
\end{equation}

It is clear that this expression decays as a function of the number of rounds $s$ only if $\frac{k}{e(sQ + k)} \cdot \frac{1}{\sqrt{p}} > 1$. Therefore, asymptotically (where $N\gg k)$, an adversary can solve at most $k < s \cdot Q \cdot e \sqrt{p}$  \pows.
\footnote{Note that this is very close to the expression we obtained in \cite{CGKSW19}, and is actually a slightly tighter bound since $c\approx 8<e^2$.}.

Next, we define the following function for the number of chained \pows, as a function of the number of consecutive rounds $s$:
\begin{equation}
    k_0(s) := s \cdot (1 + \epsilon)e Q\sqrt{p}
\end{equation}
where we recall that $\epsilon$ refers to the concentration quality.
Now, for this target function $k_0(s)$, we first observe that the probability the adversary achieves this number of chained \pows{} is:
\begin{equation}\label{eq:bound_p_k0}
   P(N, k_0) \leq \exp \left( - 2e(1 + \epsilon) s Q \sqrt{p} \ln \left( \frac{1 + \epsilon}{1 + e(1 + \epsilon)\sqrt{p}} \right) \right)
\end{equation}
Hence, we notice that $P(N, k_0)$ decays exponentially as the number of rounds $s$ increases (for any choice of $\epsilon > \frac{e\sqrt{p}}{1 - e\sqrt{p}}$).

Then, it is sufficient to obtain an honest majority condition by imposing that the honest players can achieve $k_0(s)$ chained \pows. As a result, by using Definition~\ref{def:pq_typical_exec}
, the sufficient condition becomes:
\begin{equation}
    k_0(s) \leq (1 - \epsilon)f(1 - f)s
\end{equation}
Consequently, the honest majority condition can be described by the following bound on the quantum adversarial hashing power:
\begin{equation}\label{eq:honest_maj}
    Q \leq \frac{(1 - \epsilon)f(1 - f)}{(1 + \epsilon)e \sqrt{p}}
\end{equation}

 Computing the probability $P_{\text{q}}$ of the post-quantum honest majority follows directly 
 by using Lemma~\ref{lemma:chain_pows} under the choice of parameters:
 $N = sQ$ and $k = (1 - \epsilon)f(1 - f)s$. 
\end{proof}

\begin{corollary}
The required number of rounds for safe settlement against quantum adversaries is
\begin{equation}
    s_{\text{q}} = O\left( \frac{\epsilon^2}{(1 - \epsilon)(1 - f)} \cdot s_{cl}\right)
\end{equation}
where $s_{cl}$ is the required number of rounds in the classical setting \cite{GKL15}.
\end{corollary}

\paragraph{On parallel quantum processors.} A final note is to highlight an assumption we made about the depth of the quantum computation of the adversary in $s$ rounds which is $sQ$. This assumption overestimates the power of the quantum adversary, since one can imagine an adversary that controls two or more quantum processors, so the overall depth of the computation is smaller than the total number of queries. It is known that search algorithms are not parallelizable, meaning that we have 
been over-pessimistic on the honest majority condition we derived. On the one hand, as we have mentioned in Section \ref{sec:rel}, the recording techniques cannot generalize to this setting. On the other hand, based on simple calculations of restricted parallel quantum adversaries, the factor $p^{-1/2}$ remains, reduced at best by a constant factor. Thus, qualitatively, our results persist in that case.

\begin{acknowledgements}

PW acknowledges support by the UK Hub in Quantum Computing and Simulation, part of the UK National Quantum Technologies Programme with funding from UKRI EPSRC grant EP/T001062/1. AC acknowledges support from %the French National Research Agency through the project ANR-17- CE39-0005 quBIC. 
the National Science Foundation grant CCF-1813814 and from the AFOSR under Award Number FA9550-20-1-0108.
FS was supported by the US National Science Foundation grants CCF-2042414 and CCF-2054758 (CAREER). The work of JG was supported by National Science Foundation grants CNS-2001082 and CNS-2055694.

\end{acknowledgements}

\bibliographystyle{quantum}

\bibliography{biblio}

%\onecolumn\newpage

\appendix

%%%%%%%%%%%%%%%%%%%%%%%%%%%%%

% APP A. COMPARISON QUERY

%%%%%%%%%%%%%%%%%%%%%%%%%%%%%

%%%%%%%%%%%%%%%%%%%%%%%%%%%%%%%%%%%%%%%%%%%%%%%%%%%%%%%%%%%%%%%%%%%%
\section{Comparison with Other Query Complexity Work \label{sec:comparison}}
%%%%%%%%%%%%%%%%%%%%%%%%%%%%%%%%%%%%%%%%%%%%%%%%%%%%%%%%%%%%%%%%%%%%

In this section we compare three different approaches to the query complexity analysis of the \mrsearch{} problem, in terms of their performance (not the techniques to derive the bounds). We will denote the three different bounds of the works we are comparing as follows:
\begin{itemize}
    \item $\sf{Gen_1}$ - analysis performed in the work of \cite{HM21} leading to success probability $\bar{p}_k^N$ ; 
    \item $\sf{Gen_2}$ - analysis performed in the current work leading to success probability $p_k^N$;
    \item $\sf{NonS}$ - analysis in our previous work leading to success probability $P_{\sf{NonS}}$;
\end{itemize}
We note that $\sf{Gen_1}$ and $\sf{Gen_2}$ deal with the most general (powerful) quantum adversaries, while in $\sf{NonS}$ the adversary is restricted in treating the number of queries to the oracle as a classical variable. The comparison can be summarized as follows:

\begin{table}[H]
\begin{center}
\setlength\tabcolsep{4.0pt}
 \begin{tabular}{||c | c | c | c ||} 
 \multicolumn{1}{c}{}  &   \multicolumn{1}{c}{$\sf{NonS}$} &   \multicolumn{1}{c}{$\sf{Gen_1}$} &   \multicolumn{1}{c}{$\sf{Gen_2}$} \\ 
 \hline\hline
$\begin{array} {lcl} \text{Quantum} \\ \text{Adversary} \end{array}$ & Restricted & Most General & Most General\\ 
 \hline
$\begin{array} {lcl} \text{Success} \\ \text{Probability} \end{array}$ & $\begin{array} {lcl} \text{ For } k = (1 +\epsilon) \sqrt{cp}N \\ P_{\sf{NonS}} := e^{-(cp)^{\frac{2}{3}} N \cdot g_2(\epsilon)}\end{array}$ & $\bar{p}_k^N \leq 2 \left(\frac{8e N\sqrt{p}}{k}\right)^k + \left(\frac{1}{2}\right)^k$ & $p_k^N \leq \frac{2(1 - p)}{\pi k} \cdot \left(\frac{eN\sqrt{p}}{k}\right)^{2k}$ \\
 \hline
$\begin{array} {lcl} \text{Expected} \\ \text{Optimal\footnotemark} \end{array}$ & $\sqrt{cp} N$ & $8e\sqrt{p}N$ & $e\sqrt{p}N$ \\
\hline
$\begin{array} {lcl} \text{Convergence} \\ \text{Around}\\ \text{Optimal} \end{array}$ & $\exp\left(-N\cdot O(p^{2/3})\right)$ & $ \exp \left(-N\cdot O(p^{1/2})\right)$ & $\exp \left(-N\cdot O(p^{1/2})\right)$ \\
\hline
\end{tabular}
\end{center}
 \caption{Analysis of the \mrsearch{} Problem}
\label{tab:comparison}
\end{table}
\footnotetext{Note, that by ``Expected Optimal'', we mean the value of $k$ above which the given expressions begin to converge (decay). To see how quickly the tails converges, we simply need to consider $k=(1+\epsilon)\times \text{Expected Optimal}$.}
\noindent

The first thing to observe is that for the analysis of $\sf{Gen_1}$ and $\sf{Gen_2}$, their results can be directly compared in terms of probability of success as function of the number of available queries and size of solution for the \mrsearch{} problem. Hence, we can notice from the two functions $\bar{p}_k^N$ and $p_k^N$, that our current work provides a tighter bound (besides the extra term $1/2^k$, their function has for the first term exponent $k$, while our expression has almost the same term with exponent $2k$). Hence, we remark that while the result for expected optimal length is only a constant better in our case compared to $\sf{Gen_1}$, the bound on the probability of success we obtained is much stronger than the one obtained in $\sf{Gen_1}$ (roughly our bound is the square of the bound in $\sf{Gen_1}$).

Second thing to note is that from the form of the success probability we can deduce the smallest value of $k$, above which each expression decays. This can be viewed as a bound on the expected optimal quantum strategy, since these bounds rule-out any larger value of $k$ (on average). Here we see that the best (tightest) bound on the best average quantum strategy is given by this current work and next is our previous work, with looser bound given in $\sf{Gen_1}$: $e\sqrt{p}N<\sqrt{cp}N<8e\sqrt{p}N$. Recall that $c$ is given by bounds on the query complexity of standard search and is known to be $c \sim 8$, making our new and old results very close.

A direct consequence is that the value of $k$ that $\sf{Gen_1}$ starts decaying is eight times bigger than the value that $\sf{Gen_2}$ starts decaying. This means that $\sf{Gen_1}$ overestimates the capabilities of what a quantum party (adversary) can achieve eightfold compared to $\sf{Gen_2}$.
In the language of the backbone Bitcoin protocol, this means that $\sf{Gen_1}$ leads to much worse honest majority.

The third thing to consider, is the ``speed'' of convergence to zero, as each of the cases considered moves away from its own average value. Here we see that $\sf{Gen_2}$ and $\sf{Gen_1}$ both perform better (equally well) by having a factor of $\exp\left(-N O(p^{1/2})\right)$, than our earlier work $\sf{NonS}$ that has a factor of $\exp\left(-N O(p^{2/3})\right)$. In the language of the backbone Bitcoin protocol, this means that $\sf{Gen_1},\sf{Gen_2}$ have smaller settlement time (actually essentially the same with the classical case), while our earlier work $\sf{NonS}$ has greater settlement time by a factor of $O(p^{-1/6})$. Of course, the settlement time is determined by assuming that the adversary controls the maximum allowed, by the honest majority, computational power. Given that $\sf{Gen_1}$ requires that the honest parties must have much bigger advantage (many more queries) than in $\sf{NonS}$, in practice even if such honest majority is assumed, $\sf{NonS}$ might perform better.

%%%%%%%%%%%%%%%%%%%%%%%%%%%%%

% APP B. SIMPLIFIED BOUND

%%%%%%%%%%%%%%%%%%%%%%%%%%%%%

\section{
Simplified Bound of Theorem~\ref{lemma:bound_hard}
} \label{app:proof_corr_kbersearch}

\begin{proof}

We start with denoting the sum: $S_{k, N} := \sum_{i = 0}^k (\sqrt{1 - p})^i \binom{N}{i}$. 
Firstly we drop the $\sqrt{1 - p}^i$ factors ($\sqrt{1 - p}^i < 1$), so instead we will work with the sums: $S'_{N, k} := \sum_{i = 0}^k \binom{N}{i}$. \\
We will bound the ratio:
\begin{equation}
    R_{N, k} := \frac{S'_{N, k}}{S'_{N, k - 1}} = \frac{S'_{N, k - 1} + \binom{N}{k} }{S'_{N, k - 1}} = 1 + \frac{\binom{N}{k}}{S'_{N, k - 1}}
\end{equation}

To upper bound the ratio $R_{N, k}$, we will lower bound the sum $S'_{N, k - 1}$ with the last two terms in the sum: $S'_{N, k - 1} > \binom{N}{k - 1} + \binom{N}{k - 2} = \binom{N + 1}{k - 1}$. Hence we will upper bound the ratio as:
\begin{equation}
    R_{N, k} \leq 1 + \frac{\binom{N}{k}}{\binom{N + 1}{k - 1}} = \frac{k(N + 1) + (N - k + 1)(N - k + 2)}{k(N + 1)}
\end{equation}
From this we obtain that:
\begin{equation}
    R_{N, k} \leq \frac{N}{k}
\end{equation}
which holds as long as $4 \leq k \leq N$. \\
This implies that: $S'_{N, k} < \frac{N}{k} S'_{N, k - 1}$. By applying this iteratively, and by using that $S'_{N, 0} = 1$, we obtain:
\begin{equation}
    S'_{N, k} \leq \frac{N^k}{k!}
\end{equation}
Hence, as $p_k^N \leq 4 (1 - p)p^k {S'_{N, k}}^2$, we have:
\begin{equation}
    p_k^N \leq 4 \left[\sqrt{1 - p} \cdot \sqrt{p}^k  \frac{N^{k}}{k!} \right]^2
\end{equation}
Finally, using Stirling approximation for: $k! > \sqrt{2\pi k} \cdot \frac{k^k}{e^k}$, we get:
\begin{equation}
    p_k^N \leq 4 \left[\sqrt{1 - p} \cdot \sqrt{p}^k  \left(\frac{Ne}{k}\right)^{k} \cdot \frac{1}{\sqrt{2\pi k}}\right]^2
\end{equation}

\end{proof}

\end{document}